\documentclass[journal]{IEEEtran}

\ifCLASSINFOpdf
\else
\fi

\usepackage[utf8]{inputenc} 
\usepackage[T1]{fontenc}    
\usepackage{hyperref}       
\usepackage{url}            
\usepackage{booktabs}       
\usepackage{amsfonts}       
\usepackage{nicefrac}       
\usepackage{microtype}      
\usepackage[most]{tcolorbox}
\usepackage{balance}
\usepackage[ruled,linesnumbered]{algorithm2e}
\usepackage{color,xcolor}
\usepackage{amsmath}
\allowdisplaybreaks
\usepackage{amssymb}
\usepackage{latexsym}
\usepackage{siunitx}

\usepackage{bm}
\usepackage{nicefrac}
\usepackage{array}
\usepackage{multirow}
\usepackage{threeparttable}
\usepackage{makecell}
\usepackage{stfloats}
\usepackage{graphicx}
\usepackage{amsthm}
\usepackage{color}
\usepackage{subfigure}
\usepackage{tabularx}
\usepackage{bm}
\usepackage{xcolor}

\newtheorem{theorem}{Theorem}
\newtheorem{lemma}{Lemma}
\newtheorem{definition}{Definition}

\newtheorem{proposition}{Proposition}

\theoremstyle{plain}
\newtheorem{problem}{Problem}

\newenvironment{fminipage}%
{\begin{Sbox}\begin{minipage}}%
		{\end{minipage}\end{Sbox}\fbox{\TheSbox}}

\DontPrintSemicolon
\SetKw{KwAnd}{and}
\SetFuncSty{textsc}
\SetKwInOut{Input}{Input\ \ \ \ }
\SetKwInOut{Output}{Output}

\def\ceil#1{\left\lceil #1 \right\rceil}
\def\defeq{\stackrel{\mathrm{def}}{=}}

\newcommand{\Solver}{\emph{Estimator}}
\newcommand{\wmax}{w_{{\max}}}
\newcommand{\wmin}{w_{{\min}}}

\newcommand\aaa{\boldsymbol{\mathit{a}}}
\newcommand\bb{\boldsymbol{\mathit{b}}}
\newcommand\ee{\boldsymbol{\mathit{e}}}

\newcommand\uu{\boldsymbol{\mathit{u}}}
\newcommand\vv{\boldsymbol{\mathit{v}}}

\newcommand\yy{\boldsymbol{\mathit{y}}}

\newcommand\XXtil{\boldsymbol{\mathit{\tilde{X}}}}
\newcommand\YYtil{\boldsymbol{\mathit{\tilde{Y}}}}
\newcommand\AAA{\boldsymbol{\mathit{A}}}
\newcommand\BB{\boldsymbol{\mathit{B}}}
\newcommand\DD{\boldsymbol{\mathit{D}}}

\newcommand\II{\boldsymbol{\mathit{I}}}
\newcommand\JJ{\boldsymbol{\mathit{J}}}
\newcommand\LL{\boldsymbol{\mathit{L}}}

\newcommand\QQ{\boldsymbol{\mathit{Q}}}
\newcommand\PP{\boldsymbol{\mathit{P}}}
\newcommand\RR{\boldsymbol{\mathit{R}}}
\newcommand\SSS{\boldsymbol{\mathit{S}}}
\newcommand\UU{\boldsymbol{\mathit{U}}}
\newcommand\WW{\boldsymbol{\mathit{W}}}
\newcommand\XX{\boldsymbol{\mathit{X}}}
\newcommand\YY{\boldsymbol{\mathit{Y}}}

\newcommand{\rea}{\mathbb{R}}
\newcommand{\OM}{\mathbf{\Omega}}

\newcommand\calG{\mathcal{G}}

\newcommand\Otil{\widetilde{O}}

\newcommand{\kh}[1]{\left(#1\right)}
\def\abs#1{\left|#1  \right|}
\def\norm#1{\| #1 \|}
\def\calG{\mathcal{G}}

\def\eps{\epsilon}
\newcommand{\one}{\mathbf{1}}
\def\trace#1{\mathrm{Tr} \left(#1 \right)}

\newcommand{\SpGreedy}{\emph{Greedy}}
\newcommand{\FastGreedy}{\emph{FastGreedy}}
\newcommand{\ApproxC}{\emph{ApproxC}}

\DeclareMathOperator*{\argmax}{arg\,max}

\begin{document}

\title{Measures and Optimization for Robustness and Vulnerability in Disconnected Networks}
\author{Liwang~Zhu, Qi~Bao, ~Zhongzhi~Zhang,~\IEEEmembership{Member,~IEEE}
\thanks{The work was supported by the Shanghai Municipal Science and Technology Major Project  (No. 2018SHZDZX01), the National Natural Science Foundation of China (No. U20B2051),   ZJ Lab, and Shanghai Center for Brain Science and Brain-Inspired Technology. \emph{(Corresponding author: Zhongzhi Zhang.)}} 
\thanks{The authors are with the Shanghai Key Laboratory of Intelligent Information Processing, School of Computer Science, Fudan University, Shanghai 200433, China. Zhongzhi~Zhang is also with the Research Institute of Intelligent Complex Systems, Fudan University, Shanghai 200433, China, and with the Shanghai Engineering Research Institute of Blockchain, Shanghai 200433, China.   (e-mail: 19210240147@fudan.edu.cn, 20110240002@fudan.edu.cn, zhangzz@fudan.edu.cn).}
}


\maketitle

\begin{abstract}

The function or performance of a network is strongly dependent on its robustness, quantifying the ability of the network to continue functioning under perturbations. While a wide variety of robustness metrics have been proposed, they have their respective limitations. In this paper, we propose to use the forest index as a measure of network robustness, which overcomes the deficiencies of existing metrics. Using such a measure as an optimization criterion, we propose and study the problem of breaking down a network by attacking some key edges. We show that the objective function of the problem is monotonic but not submodular, which impose more challenging on the problem. We thus resort to greedy algorithms extended for non-submodular functions by iteratively deleting the most promising edges. We first propose a simple greedy algorithm with a proved bound for the approximation ratio and cubic-time complexity. To confront the computation challenge for large networks, we further propose an improved nearly-linear time greedy algorithm, which significantly speeds up the process for edge selection but sacrifices little accuracy. Extensive experimental results for a large set of real-world networks verify the effectiveness and efficiency of our algorithms, demonstrating that our algorithms outperform several baseline schemes.

\end{abstract}

\IEEEpeerreviewmaketitle
\section{Introduction}

\IEEEPARstart{N}{etworks} or graphs are a powerful tool to describe a large variety of real systems, such as power grids and computer networks, criminal organizations, and terrorist groups. Most realistic networks are often subject to natural failures or malicious attacks, which can lead to a corrosive and detrimental risk to the functioning of societies, with regard to costs, security and disruption~\cite{WaFeKoMa19,HiBaSa09}. For example, in 2017  the ‘WannaCry’ ransomware attack on  NHS network of U.K. eventually infected around 230,000 computers over 150 countries and caused damages of billions of dollars~\cite{Tr21}. Other examples include breakdowns in power grids or water supply networks, equipment failures in communicating networks, traffic and congestion in transportation networks, and so on. Thus, it is of paramount importance for a network to continue functioning when some of its components fail. This ability can be reflected by the robustness and vulnerability of a network, which has become a fundamental subject of many studies in the past years~\cite{ElKo13,OeFa21}.

 
One of the main concerns in the research of network robustness and vulnerability is to find or establish a proper measure to quantify network robustness, which constitutes the basis to design, optimize, or control the robustness and vulnerability of a network, in order to achieve a given goal.  In the past decades, numerous robustness measures have been proposed based on partial or global information. For example, clustering coefficient~\cite{FiClFeViDo04} and betweenness centrality~\cite{Ne05} are two robustness measures using partial network information, with betweenness centrality only considering the shortest paths, but neglecting the contributions of other paths, even the second shortest paths. To better describe network robustness, some measures capturing global structure are presented, such as current-flow closeness centrality based on random walks~\cite{Ne05} and network criticality based on electrical networks~\cite{BrFl05}, both of which consider the contributions of all paths with various lengths. 

Another major concern within the field of network robustness is to optimize and control the robustness through targeted interventions on network structure. In some realistic scenarios, a less robust network might be preferred, which is termed as ‘network destruction’ in~\cite{KoIsBa15}. Here, we particularly cope with the problem of breaking a given network by identifying a set of critical links, whose removal has the most significant impact on the network robustness. Note that similar strategy of deleting edges has been exploited to various practical applications related to robustness for different purposes, including weakening  terrorist networks~\cite{AzGaNa17}, blocking rumors or misinformation in social networks~\cite{ZhAdSaVuPr16}, containing the spread of  epidemic~\cite{GaMaTo05,PaCaVaVe15,HaZh17} and preventing the propagation of computer virus in computer networks~\cite{KeWh91,KeWh93,JaCr12,WaNiZhLiNi16}, among others. Thus far, the robustness optimization problems by edge removal have been still an active research subject~\cite{YiShPa22}. 

The vast majority of existing robustness metrics of a network are introduced or designed for connected networks, which often do not apply to disconnected networks or have some limitations when applied to disconnected networks~\cite{OeFa21}. As is known to us all, many real-life networks are not connected, with frequently cited examples inducing Mobile Ad hoc Networks~\cite{DaHa07} and protein-protein interaction networks~\cite{HuLiWu18}. It is thus of theoretical and practical interest to introduce a reliable metric to characterize the robustness of disconnected graphs. On the other hand, previous related optimization algorithms for optimally selecting edges to be removed no longer work for disconnected graphs. It is significantly important to design an effective and computationally cheaper approach to optimizing the robustness of a disconnected graph based on a desirable robustness measure. These two motivations inspire us to carry out this work.

Concretely, in this paper, we first propose to use the forest index~\cite{ChSh97,ChSh98,Me98,Ch08} as a reliable robustness measure for a network $\calG$, connected or disconnected. A smaller forest index corresponds to a  more robust network, and vice versa. Then, we address the following optimization problem: given an undirected graph $\calG$ with $n$ nodes and $m$ edges, and an integer $k$, how to strategically select a set $S$ of $k$ edges to delete, so that the forest index of the resulting graph is maximized.   Our main contributions are summarized as follows.
\begin{itemize}
\item[$\bullet$] We propose to utilize the forest index as a reliable measure of robustness for a disconnected graph, which overcomes the weakness of existing measures. Based on forest index, we propose a novel edge group centrality $C(S)$ for an edge group $S$, defined as the increase of forest index when the edges in $S$  are deleted.
\item[$\bullet$] We formulate the problem of minimizing the robustness of graph $\calG$ by removing $k$ edges as an attack strategy. We show that our objective function of the optimization problem  is monotone, but not submodular.  
\item[$\bullet$] We present a simple greedy algorithm  with a bounded approximation ratio that solves the problem in $O(n^3)$ time.  
We also provide a fast algorithm to greatly accelerate evaluating $C(S)$ with computation complexity $\Otil (mk\eps^{-2})$ for any $\eps>0$,  where  $ \eps$ is the error parameter  to meet the performance-efficiency trade-off and the $\Otil (\cdot)$ notation suppresses the ${\rm poly} (\log n)$ factors.
\item[$\bullet$] We confirm the  performance of the proposed algorithms by executing extensive experiments over real-world networks of different scales, which demonstrate that our algorithms are efficient and effective, outperforming several baseline algorithms for deleting edges. 
\end{itemize}




\section{Related work}

Here we briefly review related works in terms of robustness measures, edge removal strategies for robustness optimization, and network design problems.

Over the last years, various robustness metrics have been reported in the literature~\cite{ToPrTsElFaCh10,ElKo13,FrYaKu21}, which can be roughly classified into two categories: measures based on local information and measures based on global information. 

Local robustness measures assess network robustness using relatively less structure information, with most being local information, such as degree, triangles, shortest paths, and motifs. Moreover, many local measures concentrate on the network centrality, including degree centrality, betweenness. Frequently used local measures are network diameter~\cite{NgEf06}, edge clustering coefficient~\cite{FiClFeViDo04}, degree centrality~\cite{Li78}, betweenness centrality~\cite{Ne05} of a node or edge, and so on. The diameter~\cite{NgEf06} of a network represents the longest shortest path between all node pairs. The clustering coefficient~\cite{FiClFeViDo04} of an edge reflects  the number of triangles containing the edge. The degree centrality~\cite{Li78} of a node characterizes the number of its neighbors. The betweenness centrality~\cite{Ne05} of a node or edge  measures the fraction of the shortest paths passing through the node or edge.

For those network robustness measures based on global information, they incorporate massive information of the entire graph, much of which is global even complete. Global robustness measures include the number of spanning  trees~\cite{WeGuViMe04, BaHo09}, Kirchhoff index~\cite{ElKo13,GhBoSa08, LiZh18, YaMoQiWe19,OeFa21}, the spectrum radius of the adjacency matrix~\cite{ToPrTsElFaCh10} or non-backtracking matrix~\cite{ LiChZh17, LiZh19, ZhZhCh21}, and the Randic index~\cite{DeMeRoSaVa18,BoXi19,KiKiLaPh16}.  If a network has a larger number of spanning trees, we say it is more robust.   The Kirchhoff index of a graph is defined as the average effective resistance between all pairs of nodes in the graph, with a small Kirchhoff index corresponding to a more robust graph. The spectrum radius is the largest eigenvalue or the leading eigenvalue of a matrix concerned: the smaller spectrum radius, the more robust the network. The Randic index quantifies the network robustness as the average square difference of the eigenvalues for the normalized Laplacian matrix from their mean~\cite{DeMeRoSaVa18,BoXi19,KiKiLaPh16}: the lower the Randic index, the more robust the network. Furthermore, based on the eigenvalues for the normalized Laplacian matrix, the Kemeny constant~\cite{XuShZhKaZh20, ZhXuZh20} is leveraged to measure the network vulnerability~\cite{LiWaBuCaSh21}.

Compared with local measures, global measures provide more comprehensive insights into network robustness, especially for some real scenarios such as air transportation networks~\cite{YaMoQiWe19}. However, existing local and global robustness measures, have one or more of the following deficiencies: capturing only partial structure information, being difficult to compute, applying only to connected graphs, and non-monotonically changing by edge modifications. For example, previous global robustness measures including the aforementioned ones are not be suitable to assess the robustness of disconnected networks, in spite of the facts that disconnected networks are ubiquitous in real systems. Although the average inverse distance can be used to describe the robustness of disconnected graphs~\cite{FrYaKu21}, its exact computation takes $O(n^3)$, which is impractical for larger networks with millions of nodes. In contrast, our proposed network robustness measure---forest index--- successfully avoids the pitfalls of the existing measures.


Another body of researches related to ours are the optimization of network robustness   through link removals. Actually, edge operations are convenient and practical to improve or reduce the robustness of a network. As a consequence, a concerted effort has been devoted to optimally removing a fixed number of edges for different purposes in various fields. 
In~\cite{TsSuLiHsMy94, Ra98}, the edge removal strategy is exploited to minimize the number of spanning trees. In order to contain disease dissemination, link deletion is considered in~\cite{ZhZhCh21,VaStKuLiVaLiWa11,ToPrElFaFa12} to decrease the spectral radius of related matrices, which quantifies the epidemic threshold.  In~\cite{GaNa19}, edge deletion is applied to minimize the average inverse distance. Moreover, the edge removal strategy is also used in other scenarios for different targets, such increasing the diameter~\cite{ScBoVa87} or the single-source shortest path~\cite{LiShHu00}, minimizing the
size of the $k$-core structure~\cite{MeMaSiSi20}, reducing the total pairwise connectivity~\cite{DiXuThPaZn12, ShNgXuTh13} or the closeness centrality of a given node group~\cite{VerPrPa19}.  In contrast to prior studies, we address a novel optimization problem, and present two new algorithms, with the faster one being nearly linear.


Our optimization problem for network robustness also falls into network design problems~\cite{Me19,ChPeYiTo}, which aim to optimize a certain network metric by modifying the network topology~\cite{Ch18}. Since networks have become a popular framework for modeling real systems, including VLSI, transportation, and communication systems, network design is critical to controlling realistic systems, which has been extensively applied to various aspects for different goals~\cite{Me19}. In~\cite{GhBo06},  edge addition is exploited to maximize the algebraic connectivity (the least non-zero eigenvalue of Laplacian matrix)  of a graph, which measures the extent of connectedness of the graph. In~\cite{LiPaYiZh20}, the edge addition operation is performed to maximize the number of spanning trees in a connected graph.  In~\cite{SoArAmPrAn18}, the edge addition strategy is adopted to improve node group centrality for coverage and betweenness of a  graph. In~\cite{ToPrTsElFaCh10} and~\cite{ChToPrElFaFa16,ChXuLeDuTaChPr17}, both node-level and edge-level manipulation strategies are proposed to optimize the leading eigenvalue of a network, which is a key network connectivity measure.  In~\cite{ChPeYiTo21}, the fundamental limits about the
hardness and the approximability of network connectivity optimization problems are studied.  Finally, in~\cite{ChHeBlTo17}, the measures and optimization related to connectivity for multi-layered networks are studied.
\section{Preliminaries}
This section collects basic concepts and relevant tools  to facilitate the description of our problem and the development of our greedy algorithms. Table~\ref{Notation} lists the main notations we use throughout the paper.
	\begin{table}
		\centering
		\caption{Notation Explanations.}\label{Notation}
		\begin{tabular}{ll}
			\toprule
			Notation & Definition and Description  \\
			\midrule
			$\calG$ & An undirected weighted graph \\
			$V,E,w$ & Node set, edge set, weight function in $\calG$\\
            $w_{\rm max}$, $w_{\rm min}$ & Maximum and minimum weight among all edges in $E$\\
            $\lambda_i$ & The $i$-th smallest eigenvalue of Laplacian matrix $\LL$\\
            $N(u)$ & The set of neighbours of node $u$\\
            $S$ & A set of edges to be deleted from $E$\\
            $k$ &The number of edges in $S$\\
            $C(S)$ & Increase of forest index  when edges in $S$ are deleted \\
            $\ee_i$  & The $i$-th standard basis vector \\
            $\one$, $\JJ$   & Vector and matrix  with all entries being ones\\
            $\mathbf{0}$, $\mathbf{O}$ &Vector and matrix  with all entries being zeros\\
            $\aaa^\top$, $\AAA^\top$ &Transpose of  vector $\aaa$ and matrix  $\AAA$\\
            $\AAA \preceq \BB$ & $\BB - \AAA$ is positive semidefinite\\
            $\trace{\AAA}$ & The trace of matrix $\AAA$\\
            $\norm{\AAA}_F$ & Frobenius norm of  matrix: $\norm{\AAA}_F=\sqrt{\trace{\AAA^\top \AAA}}$\\
            $a\overset{\eps}\approx b$ & $(1-\eps) b\leq a \leq (1+\eps) b$\\
			\bottomrule
		\end{tabular}
	\end{table}

\subsection{Network and Matrices}
We consider a general network topology presented by an undirected, weighted graph $\calG= (V,E,w)$ with $|V|=n$ nodes, $|E|=m$ edges, and edge weight function $w : E \to \rea_{+}$. For a pair of adjacent nodes $i$ and $j$, we write $i\sim j$ to denote $(i,j) \in E$. The set of neighbours $N(i)$ of node $i$ is denoted as  $N(i)=\{j| i\sim j\}$. The Laplacian matrix of $\calG$ is the symmetric matrix $\LL = \DD - \AAA$, where $\AAA$ is the weighted adjacency matrix and $\DD$ is the degree diagonal matrix  of $\calG$.

If each edge of $\calG$ oriented arbitrarily,   then one can define the signed edge-vertex incidence matrix $\BB_{m\times n}$ for $\calG$. Concretely, the elements of $\BB$ are defined as follows: $b_{e v}=1$ if node $v$ is the head of edge $e$, $b_{ev}=-1$ if node $v$ is the tail of edge $e$, and $b_{ev}=0$ otherwise. Let $\WW_{m\times m}$ be the diagonal edge weight matrix of graph   $\calG$, with $\WW_{ee} = w_e$.  Then, the Laplacian $\LL$ can also be represented as $\LL = \BB^\top \WW \BB$, which means that $\LL$ is singular and positive semidefinite with a unique zero eigenvalue if graph $\calG$ is connected. 
 
Let $0=\lambda_0 \leq \lambda_1\le\lambda_2\le\cdots\le\lambda_{n-1}$ be the $n$ eigenvalues of  $\LL$ for a  graph $\calG$, and $\uu_i$ be the corresponding  orthogonal eigenvectors. Then,  $\LL$ has an eigendecomposition of form $\LL=\UU \Lambda \UU^\top=\sum_{i=0}^{n-1} \lambda_i \uu_i\uu_i^{\top}$ where $\Lambda=\textrm{diag}\left(\lambda_0,\lambda_1,\lambda_2,..,\lambda_{n-1}\right)$ and $\uu_i$ is the $i$-th column of matrix $\UU$.  Let $\lambda_{\max}$ and  $\lambda_{\min}$  be, respectively, the maximum and nonzero minimum  eigenvalue of  $\LL$. Then,  $\lambda_{\max}= \lambda_{n-1}\leq n\wmax \, $~\cite{SpSr11}, and $\lambda_{\min}\geq \wmin/ n^2 $~\cite{LiSc18}.

\subsection{Forest Matrix}
The forest matrix of graph $\calG$ is defined as $\mathbf{\Omega}=\left(\II+\LL\right)^{-1}=(\omega_{ij})_{n \times n}$. It   is symmetric and positive definite, with its eigendecomposition being $\mathbf{\Omega}=\UU \tilde{\Lambda} \UU^\top=\sum_{i=0}^{n-1} \frac{1}{\lambda_i+1} \uu_i\uu_i^{\top}$, where $\tilde{\Lambda}$ is a diagonal matrix given by $\tilde{\Lambda}=\textrm{diag}(\frac{1}{1+\lambda_0}, \cdots,\frac{1}{1+\lambda_{n-2}},\frac{1}{1+\lambda_{n-1}})$ satisfying  $1=\frac{1}{1+\lambda_0} \geq \frac{1}{1+\lambda_1}\ge\cdots\ge\frac{1}{1+\lambda_{n-1}}$.  It has been shown~\cite{ChSh97,ChSh98} that  $\mathbf{\Omega}$ is  doubly stochastic, obeying  $\mathbf{\Omega} \textbf{1}=\textbf{1}$ and $\textbf{1}^\top \mathbf{\Omega}= \textbf{1}^\top $. For any connected graph, $\omega_{ij}> 0$. Moreover, $\sum_{j=1}^{n}\omega_{ij}=1$ for $i=1,2,\ldots,n$, and $\sum_{i=1}^{n}\omega_{ij}=1$
for $j=1,2,\ldots,n$.

\subsection{Submodular Functions and Greedy Algorithms}
We give the definitions of monotone and submodular set functions. For simplicity, we use $S+u$ to denote $S \cup \{u\}$. For 
a finite set $X$, let $2^X$ be the set of all subsets of $X$.

\begin{definition}[Monotonicity]
A set function $f:2^X\rightarrow \rea$ is monotone increasing if $f(S) \le f(T)$ holds for all $S \subseteq T \subseteq X$.
\end{definition}
	
\begin{definition}[Submodularity]
\label{def:sub}
A set function $f:2^X\rightarrow \rea$ is submodular if $f(S+u) - f(S) \ge f(T+u) - f(T)$ holds for all $S \subseteq T \subseteq X$ and $u\in X\backslash T$.
\end{definition}
Many network topology design problems can be formulated as maximizing a monotone submodular set function over a $k$-cardinality constraint. Formally the problem can be described as follows: 

find a subset $T^\ast$ satisfying $T^\ast\in\argmax_{|T|=k}f(T)$ where $f$ is a non-decreasing submodular set function. 

It is inherently a combinatorial problem, which means that the exhaustive search  takes exponential time to obtain the optimal
solution. Thus, the brute-force method is computationally intractable even for searching the optimal set with medium $k$. However, utilizing the diminishing property, the greedy approach of iteratively selecting the most promising elements  is known to enjoy a  guaranteed $(1-1/e)$ approximation ratio for submodular maximization problems~\cite{NeWoFi78}.

	
	
Nonetheless, for network topology design problems, there are still a major class of important functions lacks submodularity. To address this concern, a recent work~\cite{BiBuKrTs17}    provides a tight approximation guarantee of $\frac{1}{\alpha}\kh{1-e^{-\alpha\gamma}}$ for non-submodular functions, where $\alpha$ and $\gamma$ are the submodularity ratio and the generalized curvature, respectively. $\alpha$ and $\gamma$ capture the deviation of a function from being submodular and supermodular, respectively. For consistency, we introduce these two quantities below. Before doing so, we define a new quantity $\Theta_S(T)\defeq f(S\cup T)-f(T)$ to denote the marginal gain of the set $S\subseteq X$ with respect to the set $T\subseteq X$.
	
\begin{definition}[Submodular Ratio]
The submodular ratio of a non-negative set function $f$ is the largest $\gamma\in \rea_{+}$ such that for two arbitrary sets $S\subseteq X$ and $T\subseteq X$,
		\begin{equation}
		\label{eq:gamma}
		\sum_{i \in S\backslash T} \Theta_{i}(T)\ge \gamma\Theta_{S}(T).
		\end{equation}
	\end{definition} 
	
\begin{definition}[Curvature]
		The curvature of a non-negative set function $f$ is the smallest $\alpha\in\rea_{+}$ such that for  two arbitrary sets $S\subseteq X$ and $T\subseteq X$, and any element $j\in T\backslash S$,
		\begin{equation}
		\label{eq:alpha}
		\Theta_j(T\backslash j \cup S)\ge(1-\alpha)\Theta_j(T\backslash j).
		\end{equation}
	\end{definition}
	
 For a non-decreasing function $f$, it has been proved~\cite{BiBuKrTs17} that the  submodular ratio $\gamma\in[0,1]$ with $\gamma=1$  if and only if $f$ is a submodular function; and that the curvature $\alpha\in[0,1]$ with $\alpha=0$  if and only if $f$ is a supermodular function. Using these two quantities, it is possible to derive a performance guarantee for the  greedy strategy applied to a larger class of optimization problems.  
\begin{theorem}[\cite{BiBuKrTs17}]
		\label{th:subgg}
		  The greedy algorithm enjoys an approximation ratio of at least $\frac{1}{\alpha}\kh{1-e^{-\alpha\gamma}}$ for the problem of maximizing a non-negative non-decreasing set function $f$ with submodularity ratio $\gamma$ and curvature $\alpha$.
\end{theorem}

\section{Forest Index as a Measure of Robustness }\label{S4}

In this section, we adopt a spectral measure in graph theory, called forest index, to measure the robustness of disconnected networks. For consistency, we start with introducing the forest distance and the forest index. 

\subsection{Forest Distance and Forest Index}

It has been shown~\cite{ChSh97} that the elements of $\mathbf{\Omega}$ satisfy the triangle inequality: $\forall i,j,k=1,2,\ldots,n$, $\omega_{ij}+\omega_{ik}-\omega_{jk}\le \omega_{ii}$. Thus, one can naturally define a distance metric, called \textit{forest distance}.
\begin{definition}[Forest Distance]
 For an undirected graph $\calG(V,E,w)$, the forest distances $\rho_{ij}$ between a pair of nodes $i$ and $j$ is defined as 
 \begin{align*}
 \rho_{ij} \triangleq \bb_{ij}^\top \OM \bb_{ij}=\omega_{ii}+\omega_{jj}-2\omega_{ij}.
 \end{align*}
\end{definition}
In contrast to the standard geodesic  distance, forest distance ideally considers alternative paths of different lengths for all pairs of nodes~\cite{ChSh97}. For any  pair of nodes $i$ and $j$ in graph $\calG$, $0\le\rho_{ij}\le 2$ with the extreme cases that $\rho_{ij}=0$ if and only if $i=j$ and that $\rho_{ij}=2$ if and only if $i$ and $j$ are two distinct isolated nodes. On the basis of forest distance $\rho_{ij}$, we can further define the forest index, $\rho$ or $\rho(\calG)$, of a graph $\calG$.
\begin{definition}[Forest Index]
 For an  undirected graph $\calG(V,E,w)$,  its forest index $\rho(\calG)$ is defined  as the sum of forest distances over all node pairs. Namely,
\begin{align*}
\rho(\calG)\triangleq \sum_{i<j}\rho_{ij}=n\sum_{i=1}^{n-1}{\frac{1}{1+\lambda_i}}=n\trace{\OM}-n.  
\end{align*}
\end{definition}

\subsection{Robustness Interpretation}

In this subsection, we show that the forest index captures some desired properties, which make it a reliable robustness metric for a network. Intuitively, the forest index explicitly characterizes the redundancy of node-to-node paths. Moreover, the forest index also captures the information of number of walks of different length between all node pairs~\cite{ChSh97,ChSh98}. Thus, it is a global measure, including the contributions of both shortest paths and non-shortest paths.    In addition, the forest index not only is a  distance metric, but also has clear physical and structural meaning.  In particular, it explicitly characterizes the spanning forests 
of a network. For consistency, below we introduce some notions about  spanning forests.

For a graph  $\calG=(V,E,w)$, a subgraph $\mathcal{H}$ is a graph whose node sets are subsets of $V$ and edge sets are subsets of $E$, respectively. A subgraph $\mathcal{H}$ is spanning if $\mathcal{H}$ contains all the vertices of  $\calG$. A spanning forest of $\calG$ is a spanning subgraph of $\calG$ that is a forest. A spanning rooted forest of $\calG$ is a spanning forest of $\calG$ with a single node (a root) marked in each tree. For a subgraph $\mathcal{H}$ of graph $\calG$, the product of the weights of all edges in $\mathcal{H}$ is referred to the weight of  $\mathcal{H}$, denoted as  $\varepsilon(\mathcal{H})$. If $\mathcal{H}$ has no edges, its weight is set to be $1$. For any nonempty set $S$ of subgraphs, we define its weight $\varepsilon(S)$ as $\varepsilon(S) =\sum_{\mathcal{H} \in S}  \varepsilon(\mathcal{H})$. If $S$ is empty, we set its weight to be zero~\cite{ChSh97,ChSh98}.

Let $\Gamma$ be the set of all spanning forests of graph $\calG$ and $\Gamma_{i j}$ the set of those spanning forests of $\calG$ where node $j$ is in  the tree rooted at node $i$. It has been proved in~\cite{ChSh97,ChSh98} that the element $\omega_{ij}$ of forest matrix $\OM$ satisfies $\omega_{ij}=\varepsilon(\Gamma_{i j})/ \varepsilon(\Gamma)$ where $\varepsilon(\cdot)$ denotes the weight function  of corresponding spanning forests. In the case that every edge in $\calG$ has unit weight, $\varepsilon(\cdot)$ is equal to the number of spanning rooted forests. Since $\OM$ is doubly stochastic, $\omega_{ij}$ can be explained as the fraction of the connectivity of pair $(i,j)$ in the total connectivity of $i$ (or $j$ ) with all nodes~\cite{ChSh98}. Furthermore, $\omega_{ij}$ can be interpreted a the proximity between nodes $i$ and $j$~\cite{ChSh97}: the smaller the value of $\omega_{ij}$, the ``farther” $i$ from $j$. For an arbitrary pair of nodes $i$ and $j$ in $\calG$, $\omega_{ij}\geq 0$ with equality if and only if there is no path between $v_i$ and $v_j$~\cite{Me97}. 

The number of spanning forests constitutes  robustness indices for both disconnected and connected graph~\cite{ChSh97,ChSh98} while the number of spanning trees applies only to connected graph. Since forest index is expressed in terms of the entries of forest matrix $\OM$, which encapsules the information of number of spanning forests in a contracted form, it is thus a global measure of robustness.  Moreover, forest index is closely related to network connectivity. For any graph $\calG$ with $n$ nodes, $\frac{n(n-1)}{n+1}\leq \rho(\calG)\le n(n-1)$, with equality  $\rho=\frac{n(n-1)}{n+1}$ if and only if  $\calG$ is the complete graph ("best connected"), and equality $\rho=n(n-1)$ if and only if $\calG$ is the empty graph consisting of $n$ isolated nodes ("poorest connected"). Thus, it also captures the overall connectivity of a graph, with a lower value corresponding to better connectivity.  


Forest index ideally capture three desirable properties. First, forest index reflects the global connectivity and communicability in the network through alternative paths, which closely relate to robustness. It combines the robustness with network topology, graph spectra, and dynamical properties. Second, forest index is able to capture the variation of robustness sensitively even for disconnected networks where network criticality and algebraic connectivity fail to. Finally, forest index shows strict monotonicity respect to the addition/deletion of edges~\cite{ChSh97,ChSh98}, which agrees with intuition while algebraic connectivity does not show such monotonicity. These desirable properties motivate us to consider index as a reliable robustness metric.

\section{Problem Formulation}
Given an undirected weighted network $\calG=(V, E, w)$, one could make attacks on the network by deleting a set of essential edges in $S$, with the goal of breaking down the network to the greatest extent. As a result of the attack, the forest distance between any pair of vertices will increase and the same holds for the forest index, as we will show later. Then the following  problem  arises naturally: how to optimally select and delete a subset $S$ of $E$ subject to a cardinality constraint, so that the forest index of the resultant graph is maximized.

In order to tackle this problem,  we first propose a centrality measure, \emph{forest edge group centrality} (FEGC), to capture how the forest index changes under the deletion of a group of edges from the network. Then, based on FEGC, we propose an optimization problem serving as a guideline for selecting $k$ most important edges. 

\subsection{Forest Edge Group Centrality}\label{S1}
We start with establishing the notion of forest edge group centrality (FEGC). In our subsequent analysis, we will use the following notation to improve readability. We use $\calG-S$ to denote the resulting graph obtained by deleting the edges in $S$ from $E$, i.e., $\calG-S=(V, E\backslash S,\overline{w})$, where $\overline{w}:E\backslash  S \to \rea_{+}$ is the new weight function.
\begin{definition}[Forest Edge Group Centrality]
	\label{def:cenP}
	Let $\calG=(V, E, w)$ be an undirected weighted  graph and $S$ be a subset of edge set $E$. Then, the forest edge group centrality $C(S)$ for an edge group $S$ is defined as the increase of the forest index of the graph $\calG$ when all edges in $S$ are deleted from $E$, that is,
	\begin{equation}\label{eq:C}
	C(S) = \rho(\calG-S) - \rho(\calG),
	\end{equation}
where $\rho(\calG-S)$ is forest index of the resultant graph $\calG-S$.
\end{definition}
  
For notational simplicity,  we use $C(e)$ and $\calG-e$, respectively, to denote $C(\{e\})$ and $\calG-\{e\}$ hereafter. In the case that $S$ only consists of a single edge $e$, $C(e)$ is in fact a metric characterizing the importance of edge $e$. It should be noted that $C(e)$ not only applies to disconnected graph but also has a better discriminating power than other edge importance metrics for connected graph, as the following example shows.

\textbf{Example.}
Consider two edges $e_1$ and $e_2$ in the disturbed ring graph in Figure~\ref{fi:comp}. By intuition, $e_1$ plays an more important role than $e_2$ for the following arguments: The shortest path length between $i$ and $j$ will increase by $12$ if we delete $e_1$. On the contrary, the length of shortest paths between any node pairs will increase by at most $1$ as for removing $e_2$. However, betweenness centrality fails to distinguish the importance of these two edges, since their betweenness centrality  are equal to $16$. On the contrary, our proposed centrality verifies that $e_1$ is relatively more important than $e_2$, as $C(e_1)= 3.56$ and $C(e_2) = 1.73$. Such a case agrees with human intuition.
\begin{figure}
	\begin{center}
		\includegraphics[width=0.5\linewidth]{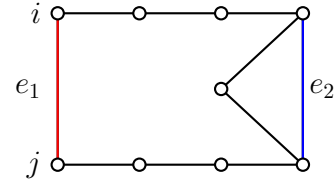}
		\caption{A disturbed ring graph. It is obtained by adding an edge $e_2$ to a 9-node ring.}
		\label{fi:comp}
	\end{center}
\end{figure}

\subsection{Maximizing  Forest Edge Group Centrality}
Having established the centrality measure of edge group, it is quite natural we transform the robustness optimization problem to the problem of maximizing FEGC subject to a cardinality constraint. We now give a mathematical formulation of Problem~\ref{prob:fomi}  in a formal way as follows.
\begin{tcolorbox}
\begin{problem}\label{prob:fomi}
	Given an undirected weighted graph $\calG= (V,E, w)$, an integer $k$, find a subset  $S\subset E$ with $|S|=k$ so that $C(S)$ is maximized. That is: 
\begin{align}
\underset{S\subset E, \,  |S|=k}{\argmax} C(S).
\end{align}
\end{problem}
\end{tcolorbox}

Such an robustness optimization problem has been the subject of many  recent papers~\cite{Ch18,ChPeYiTo}. Solutions to the problem are essential in many domains such as immunization~\cite{GaNa19}, critical infrastructure construction~\cite{BaHo09}, social collaboration mining~\cite{DaHa07}, bioinformatics analysis~\cite{HuLiWu18,ZhZhCh21}, intelligent transportation system design~\cite{ZhAdSaVuPr16}. For example,  how can we maximally break down an adversary network such as a terrorist network by cutting out some of its most important communication channels? How can we effectively contain the spread of a disease by removing some crucial links in the contagion
network? Compared with designing a new topology from the scratch, the operation of deleting edges such as cutting down airlines, hyperlinks, communication channels or social links is a cheap and viable manner to improve or enhance certain desired metrics of a network.
 

While our paper concentrates on decreasing the robustness of a network, it is worth mentioning that our techniques may also help fortify the robustness of the network. In infrastructure networks such as power grids and transportation systems, the full functioning of the systems is strongly dependent on the connectivity of the underlying networks~\cite{BaHo09}. Our approach could help the maintenance team to identify critical facilities and transmission lines whose failure would sabotage the connectivity of the entire network, so that precaution and protection measures can be implemented proactively~\cite{AzGaNa17,ChPeYiTo}.
\section{Problem Characterization and a Na\"{\i}ve Greedy Approach}
In this section, we first study the characterization of Problem~\ref{prob:fomi}.  Then, we provide a bounded approximation  greedy algorithm.  

We start with clarifying the theoretic challenges of Problem~\ref{prob:fomi}. The main computation obstacle for the Problem~\ref{prob:fomi} involves two dimensions. On one hand, searching for the best edge subset delivering the maximum of the objective is inherently a combinatorial problem, which is computationally infeasible to solve in a na\"{\i}ve brute-force manner. On the other hand, assessing the impact of a given subset of edges upon the objective involves cubic-time matrix inversion. To be specific, for each candidate edge set $S$ coming from $\tbinom{|E|}{k}$ possible subsets, we need to calculate the forest index of the resultant graph, leading to an exponential complexity $O\big(\tbinom{|E|}{k}\cdot n^3\big)$.

\subsection{Monotonicity and Non-Submodularity}
To tackle the exponential complexity, we resort to greedy heuristics. Firstly, we find that the objective function is monotone respect to the edge set $T$. However, we observe that, different from the case in other standard  optimization settings, neither submodularity nor supermodularity holds for our objective function. We present above results by the following two propositions.

\begin{proposition}[Monotonicity]\label{th:mono}
$C(S)$ is a monotonically increasing function for the edge set $S$. In other words,  for any subsets $S\subseteq T\subseteq E$, one  has $C(S) \leq C(T).$
\end{proposition}

 In order to explore the properties of objective funtion, we need to understand how the deletion of a single edge $e$ changes the forest index of the network. To be specific, the variation of the forest index under the perturbation of a single edge can be expressed by the following lemma.
\begin{lemma}\label{lem:dpedc}
	For a candidate edge $e\in E$ connecting node $u$ and $v$ with vector $\bb_e=\ee_u-\ee_v$, one obtains
	\begin{equation}\label{eq:dpdec}
	C(e)  =\frac{n w_e}{1-w_e\rho_{e}}\bb_{e}^\top \OM^2\bb_{e} \geq 0.
	\end{equation}
\end{lemma} 
\begin{proof}
The deletion of an edge could be viewed as a rank-$1$ correction of the Laplacian $\LL$. That is, if we perturb the network by removing  an edge $e$ with weight $w_e$, we obtain the Laplacian $\LL-w_e\bb_e\bb_e^\top$. The Sherman-Morrison formula~\cite{Me73} for the matrix-inverse leads to:
	\begin{align*}
	& \kh{\II+\LL - w_e\bb_e\bb_e^\top}^{-1}\nonumber 
	= (\II+\LL)^{-1} + \frac{ w_e\OM \bb_e \bb_e^\top \OM}{1 - w_e \bb_e^\top \OM \bb_e}.
	\end{align*}
Following the definitions of $C(e)$ in Eq.~(\ref{eq:C}), we can immediately obtain  
\begin{align*}
C(e)&=\rho(G-e)-\rho(\calG)\\
&=n\trace{\OM + \frac{ w_e\OM \bb_e \bb_e^\top \OM}{1 - w_e \bb_e^\top \OM \bb_e}}-n\trace{\OM}\\&=\frac{n w_e}{1-w_e\bb_e^\top \OM \bb_e}\trace{\OM\bb_e\bb_e^\top \OM} 
=\frac{n w_e \bb_{e}^\top \OM^2\bb_{e}}{1-w_e\rho_{e}},
\end{align*}
which completes the proof.
\end{proof}

Notice that the forest matrix $\OM$ is positive definite, and so is $\OM^2$. Thus, it follows that $\bb_{e}^\top \OM^2\bb_{e}>0$. On the other hand, $\rho_{e}< 1/w_e$. Hence, $C(e)> 0$, which implies the monotonicity of $C(S)$. We then show that $C(S)$ is not submodular.
\begin{proposition}[Non-Submodularity]\label{th:nonsub}
There exists an undirected weighted graph $\calG=(V,E,w)$  such that the objective function $C(S)$ defined on $E$ is not a submodular function of $S$.
\end{proposition}

To show the non-submodularity of the function concerned, consider the graph $\calG$ shown in Figure~\ref{fi:nosub} with Laplacian
\begin{equation*}
    \LL=\begin{pmatrix} 3 & -1 & -1 & -1 \\ -1 & 2 & -1 & 0 \\ -1 & -1 & 2 & 0 \\ -1 & 0 & 0 & 1 \end{pmatrix}.
\end{equation*}
To note non-submodularity,  set $S=\emptyset$, $T=\{e_{2}\}$ and $u=e_{1}$. Then,  we have
\begin{small}
\begin{align*}
C(S+u)-C(S)
=2.2 < 2.2381 = C(T+u)-C(T),
\end{align*}
\end{small}
which violates the requirement of submodularity. Thus, our objective function in Problem~\ref{prob:fomi} is non-submodular.
\begin{figure}
	\begin{center}
		\includegraphics[width=0.4\linewidth]{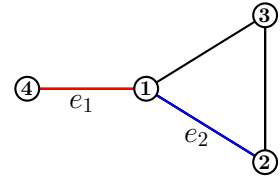}
		\caption{An undirected graph to show the non-submodularity of $C(S)$.}
		\label{fi:nosub}
	\end{center}
\end{figure}

 Greedy strategy has achieved the empirical success on a significantly larger class of non-submodular functions~\cite{BiBuKrTs17}. Despite the fact that  our objective $C(S)$ is not submodular, the greedy approach can still act as a proper manner to solve Problem~\ref{prob:fomi} with provable guarantee. We will detail the outline and bounded guarantee of the greedy approach in the next two subsections.

 \subsection{Na\"{\i}ve Greedy Approach}

  Exploiting the performance guarantee in Theorem~\ref{th:subgg} for non-submodular functions, we propose a simple greedy approach, denoted as \SpGreedy. The general approach is as follows. We first assess each candidate edge $e$ with respect to  its edge centrality $C(e)$, and then iteratively select the most promising edge for deletion  until budget is reached. 
 
The greedy approach starts with the result edge set $S$ being equal to the empty set.  Then $k$ edges are iteratively selected to the edge set from set $E\setminus S$.  At each iteration  of the greedy algorithm,  the edge $e$ in candidate set is chosen that gives the highest centrality score $C(e)$.  The algorithm terminates until $k$ edges are selected to be added to $S$.  A na\"{\i}ve greedy algorithm takes time $O(km  n^3)$, which is computationally difficult when employed on large networks. To alleviate computation burden, we  resort to matrix perturbation theory. As shown in the proof of Lemma~\ref{lem:dpedc}, with $(\II+\LL)^{-1}$  already computed, we can view the deletion of a single edge $e$ as a rank-$1$ update to the original matrix $(\II+\LL)^{-1}$, which can be obtained with a run-time of $O(n^2)$ by using Sherman-Morrison formula~\cite{Me73}, instead of inverting the matrix again in $O(n^3)$ in each loop.
 
 The above analysis leads to a  greedy algorithm $\SpGreedy(G,k)$, which is outlined in Algorithm~\ref{alg:SpG}. To begin with, this algorithm requires $O(n^3)$ time to compute the inverse of $\II+\LL$, and then it performs in $k$ rounds,  with each round mainly including two steps: computing $C(e)$ for each $e\in  E \setminus S$ (Line 4)   in $O(m n^2)$ time,  and updating  $(\II+\LL)^{-1}$ (Line 8)  in $O(n^2)$ time.  Thus,  the total running time of Algorithm~\ref{alg:SpG} is $O(n^3+kmn^2)$, much  faster than the brute-force manner.
 \begin{algorithm}
	\caption{$\SpGreedy(\calG, k)$}
	\label{alg:SpG}
	\Input{
		An undirected graph $\calG$; an integer $k \leq |E|$
	}
	\Output{
		A subset of $S \subset E$ and $|S| = k$
	}
	Initialize solution $S = \emptyset$ \;
	Compute $(\II+\LL)^{-1}$ \;
	\For{$i = 1$ to $k$}{
		Compute $C(e)$
		for each $e \in E \setminus S$ \;
		Select $e_i$ s.t. $e_i \gets \mathrm{arg\, max}_{e \in E \setminus S} C(e)$ \;
		Update solution $S \gets S \cup \{ e_i \}$ \;
		Update the graph $\calG \gets \calG- e_i $ \;
		Update $(\II+\LL)^{-1} \gets (\II+\LL)^{-1} + \frac{w_e \OM \bb_e \bb_e^\top \OM}{1 -w_e \bb_e^\top \OM \bb_e}$
	}
	\Return $S$
\end{algorithm}
\subsection{Bounded Approximation Guarantee }
 In this subsection, we analyze the proposed greedy approach in terms of the bounded approximation guarantee. According to Theorem~\ref{th:subgg}, the performance of the greedy algorithm for non-submodular function can be evaluated by its submodularity ratio $\gamma$ and curvature $\alpha$. Here, we derive the bounds for these two quantities of Problem~\ref{prob:fomi} as stated in the following lemma. 
\begin{lemma}
	\label{lem:sub}
 	The submodularity ratio $\gamma$ of set function $C(S)=\rho(\calG-S)-\rho(\calG)$ is bounded by
	\begin{equation}
		\label{eq:lowgamma}
		1 > \gamma \ge \kh{\frac{1}{1+\lambda_{n-1}(\LL)}}^2,
	\end{equation}
	and its curvature $\alpha$ is bounded by
	\begin{equation}
		\label{eq:upalpha}
		0 < \alpha \le 1- \kh{\frac{1}{1+\lambda_{n-1}(\LL)}}^2,
	\end{equation}
	where $\lambda_{n-1}(\LL)$ is the largest  eigenvalue of Laplacian matrix $\LL$.
\end{lemma}
\begin{proof}
	Let  $S$ and $T$ be two arbitrary subsets of $E$. To begin with, we provide a lower and upper bound for the marginal benefit function $\Theta_{T}(S)=C(S\cup T)-C(S)$, respectively.  We use $\Omega_{S}$ to denote the forest matrix of the resultant graph $\calG-S$. On the one hand,
	\begin{small}
	\begin{align*}
		\Theta_{T}(S) & =C(S\cup T)-C(S)  \\
		& = n\kh{\trace{\Omega_{S\cup T}} - \trace{\Omega_{S}}}\\
		& = n \sum_{i=1}^{n-1}\frac{1}{1+\lambda_i(\LL_{S\cup T})} - \frac{1}{1+\lambda_i(\LL_{S})} \\
		& \ge n \sum_{i=1}^{n-1}\frac{\lambda_i(\LL_{S}) - \lambda_i(\LL_{S\cup T})}{\kh{1+\lambda_{n-1}(\LL_{S})}\kh{1+\lambda_{n-1}(\LL_{S\cup T})}}\\
		& \ge n\frac{\trace{\LL_{S}}-\trace{\LL_{S\cup T}}}{\kh{1+\lambda_{n-1}(\LL_{S})}\kh{1+\lambda_{n-1}(\LL_{S\cup T})}} \\
		& = \frac{2n\sum_{e \in T\backslash S}w_e}{\kh{1+\lambda_{n-1}(\LL_{S})}\kh{1+\lambda_{n-1}(\LL_{S\cup T})}} 
	\end{align*}
	\end{small}
	On the other hand,
	\begin{small}
	\begin{align*}
	\Theta_{T}(S) & = n \sum_{i=1}^{n-1}\frac{\lambda_i(\LL_{S})-\lambda_i(\LL_{S\cup T})}{\kh{1+\lambda_{n-1}(\LL_{S})}\kh{1+\lambda_{n-1}(\LL_{S\cup T})}}\\
	& \le n\frac{\trace{\LL_{S}}-\trace{\LL_{S\cup T}}}{\kh{1+\lambda_1(\LL_{S})}\kh{1+\lambda_1(\LL_{S\cup T})}}\\
	& = \frac{2n\sum_{e \in T\backslash S}w_e}{\kh{1+\lambda_1(\LL_{S})}\kh{1+\lambda_1(\LL_{S\cup T})}}.
	\end{align*}
	\end{small}
	Combining  the above two bounds together, we derive the lower bound of the submodular ratio $\gamma$
	\begin{small}
	\begin{align*}
		\frac{\sum_{e\in T\backslash S}\Theta_e(S)}{\Theta_{T}(S)} 
		 \ge \sum_{e\in T\backslash S} \frac{2n w_e}{\kh{1+\lambda_{n-1}(\LL_{S})}\kh{1+\lambda_{n-1}(\LL_{S\cup T})}} \times\\ \frac{\kh{1+\lambda_1(\LL_{S})}\kh{1+\lambda_1(\LL_{S\cup T})}}{2n\sum_{e \in T\backslash S}w_e} 
		\ge \kh{\frac{1}{1+\lambda_{n-1}(\LL)}}^2,
	\end{align*}
	\end{small}
	which yields~\eqref{eq:lowgamma}.
	
	In a similar way, we derive an upper bound of the curvature $\alpha$. Let $j$ be any candidate edge in $S\backslash T$. Then, we have
	\begin{small}
	\begin{align*}
		\frac{\Theta_j(S\backslash j\cup T)}{\Theta_j(S\backslash j)} \ge \frac{2n w_j}{\kh{\lambda_{n-1}(\LL_{S\backslash j\cup T})}\kh{\lambda_{n-1}(\LL_{S\cup T})}}\times\\ \frac{\kh{1+\lambda_{1}(\LL_{S\backslash j})}\kh{1+\lambda_{1}(\LL_{S})}}{2n w_j} \ge \kh{\frac{1}{1+\lambda_{n-1}(\LL)}}^2,
	\end{align*}
	\end{small}
	which combining with the curvature definition in~\eqref{eq:alpha} completes the proof of~\eqref{eq:upalpha}. 
\end{proof}

 Lemma~\ref{lem:sub}, together with the approximation guarantee stated in Theorem~\ref{th:subgg}, leads to a performance analysis for the greedy algorithm. It is worth mentioning that in practice the greedy approach has been shown to perform often very close to the optimal solutions. 

\section{Nearly-Linear Time Approximation algorithm}
While computation time of Algorithm~\ref{alg:SpG} is significantly reduced compared with the na\"{\i}ve algorithm,  it is still  computationally difficult for large networks with millions of nodes. Specifically, as can be seen from Eq.~(\ref{eq:dpdec}), it remains questionable how to efficiently compute $C(e)$,  which requires computing the inverse of  matrix $\II+\LL$. To address this issue, we first present $C(e)$ in form of $\ell_2$ norm of vectors, which enables us to greatly facilitate the evaluation of  $C(e)$ by   several integrated strategies without needing matrix inversion. Then, we  propose an efficient approximation algorithm to solve the problem in time $\Otil (km\eps^{-2})$ where the $\Otil (\cdot)$ notation suppresses the ${\rm poly} (\log n)$ factors. We below provide its mathematical framework, which leads to an $\eps$-approximation estimator of $C(e)$, and then present the outline of $\FastGreedy$  algorithm.
\subsection{Mathematical Framework}
For Algorithm~\ref{alg:SpG}, the main difficulty lies in calculating  $C(e)$ for each candidate edge $e$.  According to Eq.~\eqref{eq:dpdec},  to evaluate $C(e)$, we need to estimate two terms $\trace{\OM\bb_e\bb_e^\top \OM}$  and $\bb_e^\top \OM \bb_e$ in the numerator and denominator,  respectively. In order to facilitate our analysis, we explicitly represent the concerned quantities in $\ell_2$ norm of vectors. Note that the forest distance $\rho_e$ between two end nodes of the deleted edge $e$ can be written in an Euclidean norm as
\begin{small}
\begin{align*}
	\rho_e=& \bb_e^\top \OM \bb_e
	= \bb_e^\top \OM (\II+\LL) \OM \bb_e  
	= \bb_e^\top \OM (\II+\BB^\top \WW \BB) 
	\OM \bb_e  \\
	=&
	\norm{\OM \bb_e}^2 +
	\norm{\WW^{1/2}\BB \OM \bb_e}^2.
\end{align*}
\end{small}
The term $\trace{\OM\bb_e \bb_e^\top \OM}$ can  be also recast as
\begin{align*}
    \trace{\OM\bb_e \bb_e^\top \OM}=&\bb_e^\top \OM^2 \bb_e=\norm{\OM\bb_e}^2.
\end{align*}
In this way,  the estimation of each part of $C(e)$ has been reduced to the calculation of the $\ell_2$ norms $\norm{\OM \bb_e}^2$ and $\norm{\WW^{1/2}\BB \OM \bb_e}^2$ of vectors in $\mathbb{R}^{m}$ and $\mathbb{R}^{n}$.  To alleviate the high computational cost of exactly computing these two $\ell_2$ norms,  we exploit the JL lemma~\cite{JoLi84,Ac01} to reduce the dimensions. JL lemma states that if we project a set of vectors $\vv_1, \vv_2, \cdots, \vv_n \in \mathbb{R}^d$  (like the columns of matrix $\OM$) onto a low $p$-dimensional
subspace spanned by the columns of a random matrix $\RR_{p \times d}$
with entries $\pm 1/\sqrt{p}$
, where $p \geq 24\log n/\eps^2$ for given $\eps$, then  the distances between the vectors in the set are nearly preserved with tolerance $1 \pm \eps$, but  the computational cost could be significantly reduced. Formally, given a random matrix $\RR$ and a vector set $\vv$,
\begin{align*}
	  (1-\eps)\norm{\vv_i - \vv_j}^2\leq	\norm{\RR \vv_i - \RR \vv_j}^2 \leq (1+\eps)\norm{\vv_i - \vv_j}^2,
\end{align*}
holds with probability at least $1 - 1/n$.

Let $\QQ_{p\times m}$ and $\PP_{p\times n}$ be two random $\pm1/\sqrt{p}$ matrices where $p=\ceil{24\log n/ (\frac{\eps}{12})^2}$. Following the JL lemma,
\begin{small}
\begin{align*}
    \norm{\WW^{1/2}\BB \OM\bb_e}^2  
    \overset{\eps/12}\approx \norm{\QQ \WW^{1/2} \BB \OM\bb_e}^2 \text{and}
    \norm{\OM\bb_e}^2 \overset{\eps/12}\approx \norm{\PP\OM\bb_e}^2
\end{align*}
\end{small}hold for every edge $e$ with probability at least $1-1/n$. 

However, direct computation of the above $\ell_2$ norms remains challenging since it involves matrix inversion, leading to a run-time of $O(n^3)$. Here, instead of computing this inverse directly, we resort to a nearly-linear time estimator~\cite{SpTe14} to solve some linear equations. Recalling that $\OM=(\II+\LL)^{-1}$, the product $\PP\OM$ is thus a solution to the linear system $\XX (\II+\LL)=\PP$. For a matrix $\XX$, we use $\XX_i$ to denote its $i$-th row vector. Then, we turn to solve a linear system of only $p=O(\log n)$ equations $(\II+\LL)\XX_i=\PP_i$, instead of $n$ systems that are required for computing $(\II+\LL)^{-1}$. Thanks to the fact that $\II+\LL$ is a symmetric, diagonally-dominant M-matrix (SDDM), the solution of each linear system  could be obtained efficiently by the SDDM linear system estimator \cite{SpTe14}, which relies on the approach of preconditioned conjugate gradients to give the unique solution $\XX_i=(\II+\LL)^{-1}\PP_i$. Formally, given an SDDM matrix $\SSS_{n\times n}$ with $m$ nonzero entries, a vector $\bb \in \mathbb{R}^n$,  and an error parameter $\delta > 0$,  the SDDM linear system estimator returns a vector $\yy = \Solver(\SSS,  \bb,  \delta)$, satisfying
\begin{align}\label{solver}
 \norm{\yy - \SSS^{-1} \bb}_{\SSS} \leq \delta \norm{\SSS^{-1} \bb}_{\SSS}
\end{align}
with  probability at least $1-1/n$,  where $\norm{\yy}_{\SSS} \defeq \sqrt{\yy^\top \SSS \yy}$. The estimator runs in expected time $\Otil (m)$. 

We below utilize the estimator in (\ref{solver}), to provide an approximations to the term $\norm{ \WW^{1/2}\BB \OM \bb_e}^2$ in~\eqref{eq:dpdec}. Specifically, we establish this estimator, which implies, for any small $\eps>0$, we could choose $\delta$ properly so that the approximation error can be bounded by $\eps$ with probability at least $1-1/n$. To proceed, let $\XX=\QQ \WW^{1/2} \BB$, $\overline{\XX}= \WW^{1/2}\BB\OM$, $\XX'=\QQ \overline{\XX}$, let $\XXtil_i=\Solver(\II+\LL, \XX_i,  \delta_1)$, and choose $\delta_1$ such that
\begin{small}\begin{align}\label{del1}
		\delta_1 \leq 
		\frac{\eps\wmin\sqrt{2(1-\eps/12)\wmin}}{64\wmax n(n+1)\sqrt{(1+\eps/12)(n\wmax+1)n}}.
		\end{align}
	\end{small} 
We then show the constructed estimator $\norm{\XXtil \bb_e}^2$ achieves an $(\eps/3)$-approximation to the  term $\norm{ \WW^{1/2}\BB \OM \bb_e}^2$ as stated in the following lemma. 
	\begin{lemma}\label{lem:appro2}
		For any  parameter $\epsilon \in (0, \frac{1}{2})$, suppose that
			\begin{align*}
		(1-\frac{\eps}{12})\norm{\overline{\XX} \ee_u}^2 \leq \norm{\XX' \ee_u}^2\leq (1+\frac{\eps}{12})\norm{\overline{\XX}  \ee_u}^2 
		   \end{align*}
	holds for any node $u \in V$ and that
			\begin{align*}
		(1-\frac{\eps}{12})\norm{\overline{\XX}  \bb_e}^2 \leq \norm{\XX' \bb_e}^2\leq (1+\frac{\eps}{12})\norm{\overline{\XX} \bb_e}^2 		
		 \end{align*}
	holds for any pair of  edge $e$ in $E$.
		Then, the following relation holds:
		\begin{align*}
		\norm{\WW^{1/2}\BB \OM\bb_e}^2=\norm{\overline{\XX} \bb_e}^2 \overset{\frac{\eps}{3}}\approx  \norm{\XXtil \bb_e}^2.
		\end{align*}
	\end{lemma}	
\begin{proof}
We observe that
 	\begin{small}
 	\begin{align*}
 	&\abs{\norm{\XXtil\bb_e}-\norm{\XX'\bb_e}}\overset{(a)}\leq  \norm{(\XXtil-\XX')\bb_e}  \overset{(b)}\leq 2\norm{\XXtil-\XX'}_F \\
 	=&2\sqrt{\sum_{i=1}^{p}\norm{\XXtil_i-\XX'_i}^2} \overset{(c)}\leq 2\sqrt{\sum_{i=1}^{p}\norm{\XXtil_i-\XX'_i}_{\II+\LL}^2}\\
 	\leq & 2\sqrt{\delta_1^2\sum_{i=1}^{p}\norm{\XX'_i}_{\II+\LL}^2} \overset{(d)}\leq 2\delta_1\sqrt{n\wmax+1}\norm{\XX'}_{F}\\
 	\leq & 2\delta_1\sqrt{n\wmax+1}\sqrt{\sum_{i=1}^{n}(1+\frac{\eps}{12})\norm{\bar{\XX} \ee_i}^2}\\
 	\leq & 2\delta_1\sqrt{n\wmax+1}\sqrt{1+\frac{\eps}{12}}\sqrt{\sum_{i=1}^{n}{\ee_i^\top \OM \ee_i}}\\
 	\leq & 2\delta_1\sqrt{n\wmax+1}\sqrt{1+\frac{\eps}{12}}\sqrt{n},
 	\end{align*}
 	\end{small}
 	where (a) and (b) follow from the triangle inequality. The inequality (c) holds since $\II \preceq \II+\LL$, and (d) holds because $\LL \preceq (n\wmax+1)\II$.
 
 	On the other hand, we  derive a lower bound of $\norm{\XX'\bb_e}^2$ as:
	\begin{align*}
	&\norm{\XX'\bb_e}^2\geq (1-\frac{\eps}{12})\norm{\overline{\XX}\bb_e}^2\\
	=&(1-\frac{\eps}{12})\bb_e^\top \OM \LL \OM \bb_e\geq (1-\frac{\eps}{12})\frac{\wmin}{\wmax^2n^2(n+1)^2}\norm{\bb_e}^2\\
	=&2(1-\frac{\eps}{12})\frac{\wmin}{\wmax^2n^2(n+1)^2}.
	\end{align*}
	Combining the above-obtained results, we further obtain
	\begin{small}
		\begin{align*}
			&\frac {\abs{\norm{\XXtil\bb_e}-\norm{\XX'\bb_e}}}{\norm{\XX'\bb_e}} \\
		\leq & \frac{2\delta_1\wmax \sqrt{n}n(n+1)\sqrt{(1+\eps/12)(n\wmax+1)}}{\sqrt{2\wmin(1-\eps/12)}}&\leq \frac{\eps}{32}.
		\end{align*}
	\end{small}
 It thus follows that
\begin{align*}
	&\abs{\norm{\XXtil\bb_e}^2-\norm{\XX'\bb_e}^2}\nonumber \\
	=&\abs{\norm{\XXtil\bb_e}-\norm{\XX'\bb_e}} \times \abs{\norm{\XXtil\bb_e}+\norm{\XX'\bb_e}}\nonumber\\
	\leq & \frac{\eps}{32}(2+\frac{\eps}{32})\norm{\XX'\bb_e}^2\leq \frac{\eps}{12}\norm{\XX'\bb_e}^2. \nonumber
\end{align*}
With the initial condition of the lemma, one gets
$\norm{\overline{\XX} \bb_e}^2 \overset{\frac{\eps}{3}}\approx  \norm{\XXtil \bb_e}^2$, which completes the proof.
\end{proof}

 The analysis on the term $\norm{\PP\OM \bb_e}^2$ can be conducted in a similar way as what follows. Let $\YY=\PP$, $\overline{\YY}=\OM$ and $\YY'=\PP \overline{\YY}$. Let $\YYtil_i=\Solver(\II+\LL, \YY_i, \delta_2)$ with 
\begin{small}\begin{align}\label{del2}
 	\delta_2 \leq
 	\frac{\eps\sqrt{2(1-\eps/12)\wmin}}{32(n\wmax+1)\sqrt{(1+\eps/12)(n\wmax+1)}}.\end{align}
\end{small} 
 Then, the following lemma provides an efficient approximation to $\norm{ \OM \bb_e}^2$.
 \begin{lemma}\label{lem:appro1}
 	For any  parameter $\epsilon \in (0, \frac{1}{2})$, suppose 
 	\[(1-\frac{\eps}{12})\norm{\overline{\YY} \ee_u}^2 \leq \norm{\YY' \ee_u}\leq (1+\frac{\eps}{12})\norm{ \overline{\YY} \ee_u}^2 \]
 	for any node $u \in V$ and 
 	\[(1-\frac{\eps}{12})\norm{ \overline{\YY} \bb_e}^2 \leq \norm{\YY' \bb_e}\leq (1+\frac{\eps}{12})\norm{ \overline{\YY} \bb_e}^2 \]
	for any pair of nodes $u$ and $v$ connecting an edge $e$ in $E$.
  	Then, the following relation holds:
 \begin{align}\label{ineq3}
 	\norm{ \OM\bb_e}^2=\norm{\overline{\YY} \bb_e}^2 \overset{\frac{\eps}{3}}\approx  \norm{\YYtil \bb_e}^2.
 \end{align}
 \end{lemma}
Based on the results obtained in  Lemmas~\ref{lem:appro2} and Lemmas~\ref{lem:appro1}, we are able to approximate the forest distance $\rho_{e}$ by $\norm{\XXtil \bb_e}^2+\norm{\YYtil \bb_e}^2$ satisfying
\begin{align*}
\rho_{e} \overset{\frac{\eps}{3}}\approx \norm{\XXtil \bb_e}^2+\norm{\YYtil \bb_e}^2.
\end{align*}

\subsection{Nearly-Linear Time Algorithm }
The above-obtained results considerably provide a provable approximation guarantee to each part of $C(e)$. Now we are ready to propose an algorithm $\ApproxC$ to approximate $C(e)$ for every edge $e$ in the candidate set $E$. The outline of  algorithm  $\ApproxC$ is shown in Algorithm \ref{alg:comp}, with performance given by the following theorem. 
\begin{theorem}\label{lem:comp}
For $\epsilon\in(0,\frac{1}{2})$, the value $\hat{C}(e)$ returned by $\ApproxC$ satisfies
\begin{align}
		\hat{C}(e) \overset{\frac{\eps}{3}}\approx C(e).
\end{align}
with probability almost $1-1/n$.
\end{theorem}
\begin{algorithm}
		\small
		\caption{$\ApproxC(\calG,\epsilon)$}
		\label{alg:comp}
		\Input{
			An undirected graph $\calG$; a real number $\epsilon\in(0,\frac{1}{2})$\\
		}
		\Output{
			$\{(e, \hat{C}(e) | e \in E\}$
			
		}
		Set $\delta_1$, $\delta_2$ according to Lemmas~\ref{lem:appro2} and~\ref{lem:appro1},
		\;
		$p \gets \ceil{24\log n/ (\frac{\eps}{12})^2}$ \;		
		Generate random Gaussian matrices
		$\PP_{p\times n},  \QQ_{p\times m}$\;
		Compute $\QQ\BB$
		by sparse matrix multiplication\;
		\For{$i = 1$ to $p$}{
			$\XXtil_i
			\gets \Solver(\II+\LL,
			\XX_i,  \delta_1)$\;
			$\YYtil_i
			\gets \Solver(\II+\LL,
			\YY_i,  \delta_2)$
		}
		
		\For{each $e\in S$}{
			compute $\hat{C}(e) =\frac{n w_e\norm{\YYtil\bb_e}^2}{1 - w_e( \norm{\XXtil\bb_e}^2 + \norm{\YYtil\bb_e}^2)}$}
		\Return $\{(e, \hat{C}(e)| e \in E\}$
	\end{algorithm}

Finally, the complete methodology $\FastGreedy(\calG, k, \epsilon)$ to solve Problem~\ref{prob:fomi} is given in Algorithm~\ref{alg:Appro}. The idea is still to greedily select one "best" edge in each round. Yet, we  exploit $\ApproxC$ to obtain the approximation of $\hat{C}(e)$ instead of directly computing $C(e)$ in each round, which only takes time $\Otil(m\eps^{-2})$. We then iteratively select the edge with the highest centrality score $\hat{C}(e)$ and update the solution set $S$ and graph $\calG$.  Thus,  the time complexity of  Algorithm~\ref{alg:Appro} is $\Otil (mk\eps^{-2})$.

\begin{algorithm}[htbp]
	\caption{$\FastGreedy(\calG, k, \epsilon)$}
	\label{alg:Appro}
	\Input{
		An undirected graph $\calG$; an integer $k \leq |E|$; a real number $0 \leq \epsilon \leq1/2$
	}
	\Output{
		$S$: a subset of $E$ and $|S| = k$
	}
	Initialize solution $S = \emptyset$ \;
	\For{$i = 1$ to $k$}{
		$\{e, \hat{C}(e) | e \in E \setminus S \} \gets \ApproxC(\calG, \epsilon)$ \;
		Select $e_i$ s.t.  $e_i \gets \mathrm{arg\, max}_{e \in E \setminus S} \hat{C}(e)$ \;
		Update solution $S \gets S \cup \{ e_i \}$ \;
		Update the graph $\calG \gets \calG- e_i$
	}
	\Return $S$
\end{algorithm}

\section{Experiments Evaluation}\label{S7}
In this section, we present numerical results to evaluate the performance and scalability of our two algorithms $\SpGreedy$ and $\FastGreedy$. For this purpose,  we design and conduct extensive experiments  on real-world networks of various types and scales to validate  the effectiveness and efficiency of our algorithms in a myriad of high-impact applications.

\begin{table*}
	\fontsize{7.0}{7.5}\selectfont
	\centering
	\caption{Running time (seconds, s) and relative error ($\times 10^{-2}$) of $\SpGreedy$ and $\FastGreedy$ for maxmizing forest index with $k=50$ and various $\eps$. }\label{table:eff}
	\resizebox{0.9\linewidth}{!}{
    \begin{tabular}{lrrccccccc}
    \toprule
\multicolumn{1}{l}{\multirow{2}{*}{Network}} & \multicolumn{1}{c}{\multirow{2}{*}{$n$}} & \multicolumn{1}{r}{\multirow{2}{*}{$m$}} & \multicolumn{4}{c}{Running time (s) for $\SpGreedy$ and $\FastGreedy$} & \multicolumn{3}{c}{Relative error ($\times10^{-2}$)} \\
\cmidrule(lr){4-7}
\cmidrule(lr){8-10}
\multicolumn{1}{c}{} & \multicolumn{1}{c}{} & \multicolumn{1}{c}{} & \multicolumn{1}{c}{$\SpGreedy$} & \multicolumn{1}{c}{$\eps=0.3$} & \multicolumn{1}{c}{$\eps=0.2$} & \multicolumn{1}{c}{$\eps=0.1$}  & \multicolumn{1}{c}{$\eps=0.3$} & \multicolumn{1}{c}{$\eps=0.2$} & \multicolumn{1}{c}{$\eps=0.1$} \\
\midrule
EmailUniv & 1133 & 5451 & 3.56 & 6.21 & 13.79 & 64.24 & 1.68 & 1.36 & 1.34 \\
Erdos992 & 5094 & 7515 & 195.26 & 56.67 & 127.27 & 541.72 & 9.63 & 8.16 & 4.12 \\
Bcspwr10 & 5300 & 8271 & 219.42 & 88.99 & 203.57 & 770.53 & 3.20 & 3.31 & 1.84 \\
Reality & 6809 & 7680 & 445.93 & 59.45 & 135.52 & 528.91 & 0.19 & 0.20 & 0.20 \\
PagesGovernment & 7057 & 89429 & 492.36 & 660.62 & 1477.58 & 5876.43 & 2.27 & 1.60 & 1.17 \\
Dmela & 7393 & 25569 & 559.98 & 221.56 & 494.50 & 1967.26 & 4.67 & 3.74 & 2.28 \\
HepPh & 11204 & 117619 & 1849.59 & 933.56 & 2102.74 & 8303.21 & 2.99 & 2.80 & 2.00 \\
Anybeat & 12645 & 49132 & 2678.16 & 430.35 & 971.64 & 3904.42 & 5.23 & 4.27 & 3.72 \\
PagesCompany & 14113 & 52126 & 3665.19 & 541.76 & 1197.31 & 4778.98 & 5.61 & 5.11 & 3.24 \\
AstroPh & 17903 & 196972 & 7344.40 & 1716.95 & 3852.60 & 15367.60 & 2.92 & 3.34 & 2.12 \\
CondMat & 21363 & 91286 & 12470.50 & 929.77 & 2093.50 & 8280.84 & 5.54 & 5.49 & 2.20 \\
Gplus & 23628 & 39194 & 16611.82 & 368.33 & 829.80 & 3281.90 & 4.19 & 4.00 & 2.98 \\
Douban$*$ & 154908 & 327162 & --- & 968.77 & 2162.07 & 8533.24 & ---  & ---  & --- \\
Gowalla$*$ & 196591 & 950327 & --- & 2993.00 & 6601.74 & 26197.87 & ---  & ---  & ---  \\
GooglePlus$*$ & 211187 & 1141650 & --- & 3933.51 & 8621.27 & 34279.85 & ---  & ---  & ---  \\
Citeseer$*$ & 227320 & 814134 & --- & 2620.55 & 5900.50 & 23530.06 & ---  & ---  & --- \\
MathSciNet$*$ & 332689 & 820644 & --- & 3277.76 & 7260.85 & 28696.81 & --- & --- & --- \\
TwitterFollows$*$ & 404719 & 713319 & --- & 2355.78 & 5177.25 & 20646.97 & --- & --- & --- \\
YoutubeSnap$*$ & 1134890 & 2987624 & --- & 13723.63 & 30466.71 & 120837.71 & --- & --- & --- \\
Lastfm$*$ & 1191805 & 4519330 & --- & 20829.48 & 46099.36 & 184029.38 & --- & --- & --- \\
Flixster$*$ & 2523386 & 7918801 & --- & 35293.35 & 78904.72 & 317482.33 & --- & --- & ---\\
\bottomrule
\end{tabular}
}
\end{table*}
\subsection{Experiment Setup}
\textbf{Datasets.}
The studied realistic networks are representatively selected from various domains,   including Criminal organizations (Crime-moreno), Epidemic contagion networks (Bio-diseasome), Biological networks (Bio-celegans), Social network (Douban, TwitterFollows) and Infrastructure networks (Inf-euroroad) and so on. The removal of edges may correspond to cutting down airlines, hyperlinks, communication channels or social links in the above-mentioned domains. The networks are publicly available in the KONECT~\cite{Ku13} and SNAP~\cite{LeSo16}. For each network,  we implement our experiments on its largest components. Table~\ref{table:eff} shows the relevant characteristics of  all networks.

\textbf{Machine and implementation.}
All algorithms in our experiments are implemented in Julia using a single thread.  In our algorithms,  we resort to the linear estimator $\Solver$~\cite{KySa16},  the Julia implementation  of  which is available on the website\footnote{https://github. com/danspielman/Laplacians. jl}.  All experiments are conducted on a machine equipped with  32G RAM and 4.2 GHz Intel i7-7700 CPU. 
The error parameter is set to $\eps= 0.3$ for the approximation algorithm $\FastGreedy$. Note that one can adjust $\eps$ to achieve a balance
between effectiveness and efficiency, where a smaller value of  $\eps$ corresponds to better effectiveness but relatively poor efficiency.

\textbf{Edge attack methods.}
The sets of $k$ edges are selected following eight different strategies: \emph{Optimum}, \emph{Random}, \emph{Betweenness}, \emph{DegProduct}, \emph{DegSum}, \emph{FastGreedy}, \emph{TopFEGC} and \emph{SpGreedy}. Note that there is no state of the art, as no existing methods can solve our optimization problem, thus we turn to  heuristic methods. \emph{Optimum} selects edge set $|T|=k$ with maximum $C(S)$ by brute-force search. \emph{Random} iteratively selects $k$ edges at random. \emph{FastGreedy} iteratively selects $k$ edges with maximum $\hat{C}(e)$  returned by algorithm $\FastGreedy$. \emph{Greedy} selects $k$ edges with maximum $C(e)$ returned by algorithm $\SpGreedy$. We also consider three classical edge-centrality-based attack strategies for
our experiments. \emph{Betweenness}, \emph{DegProduct} and \emph{DegSum} iteratively select $k$ edges with the highest betweenness~\cite{Br01}, the largest product of degrees of end-points, and the largest sum of degrees of end-points, respectively, in a greedy fashion. \emph{TopFEGC} selects the top-$k$ FEGC in one shot (not in a greedy fashion). The objective of the comparison between ~\emph{TopFEGC} and \emph{FastGreedy} is to verify the need of an iterative greedy method. 

\textbf{Evaluation metrics.}
The performance of all above-mentioned methods is evaluated with respect to the increased forest index $\Delta \rho(\calG)$ they achieve for the selected edge set. The larger $\Delta \rho(\calG)$ a method achieves, the more effective it is.
\subsection{Effectiveness Comparison}
We start with evaluating the effectiveness of proposed algorithms,  by comparing them with both the optimum solutions \emph{Optimum} and  random scheme \emph{Random}. To this end,  we  execute experiments on four small realistic networks: Tribes with $16$ nodes and $58$ edges, Southernwomen with $18$ nodes and $64$ edges, Karate with $34$ nodes and $78$ edges and Dolphins with $62$ nodes and $159$ edges. Note that these networks are relatively small, we are thus able to obtain the optimal set of edges. For each case,  we attack the network by removing $k=1,2,\ldots,6$ edges. Figure~\ref{ComOpt1} shows how the forest index is affected by the deletion of an incrementally larger set of edges. Each curve in the plot illustrates the gain by a different algorithm. We can draw the following observations. On one hand, the values returned by our two greedy algorithms and the optimum solution are almost the same so that the three curves are overlapped, demonstrating that our greedy algorithms perform much better than the theoretical guarantee. On the other hand, both of our algorithms perform significantly better than  the random scheme.
\begin{figure}
	\centering
	\includegraphics[width=\linewidth]{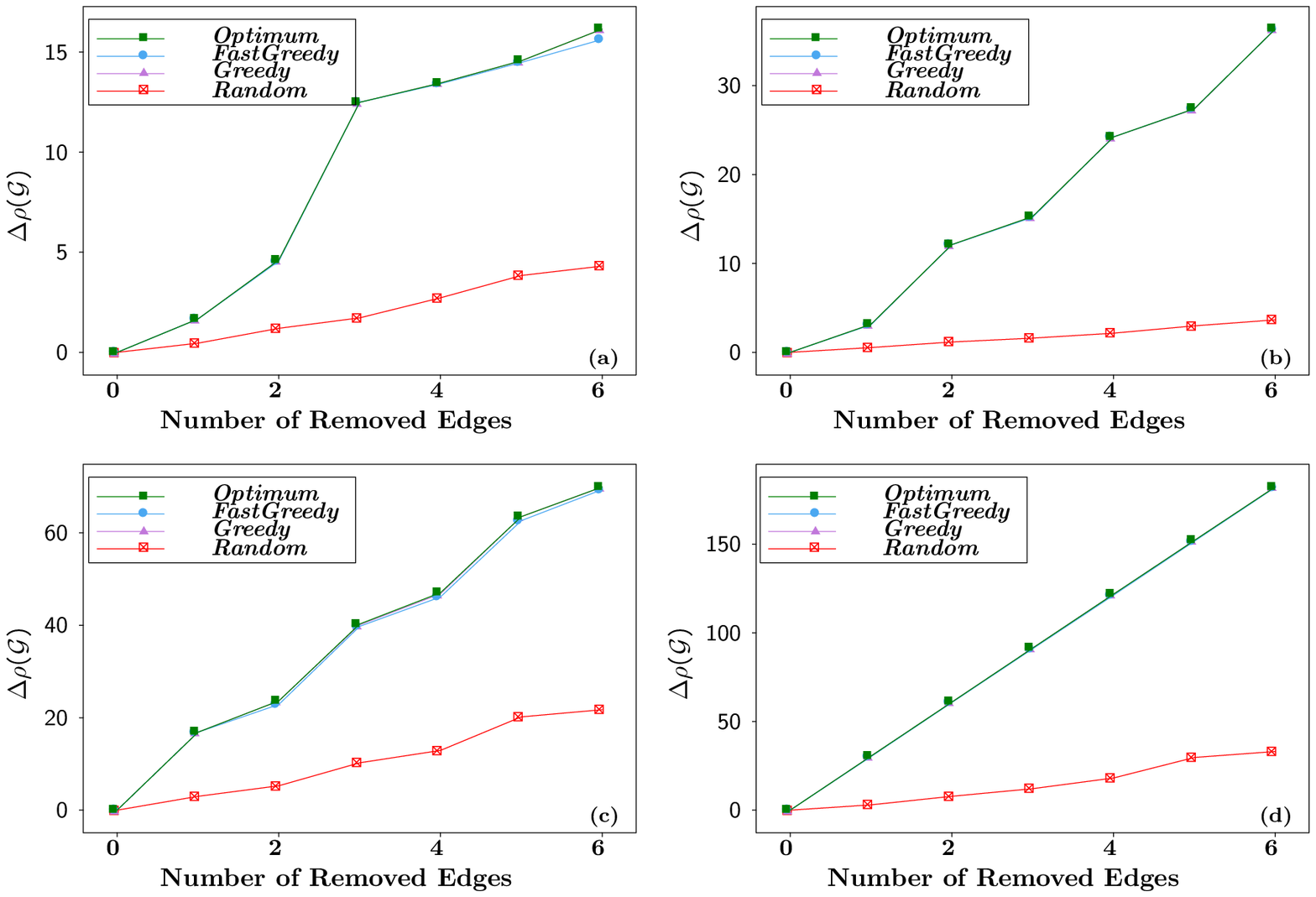}
	\caption{Increased forest index of four methods for edge attack on datasets (a) Tribes, (b) Southernwomen, (c) Karate and (d) Dolphins for increasing value of $k$.  }\label{ComOpt1}
\end{figure}

\begin{figure}
	\centering
	\includegraphics[width=\linewidth]{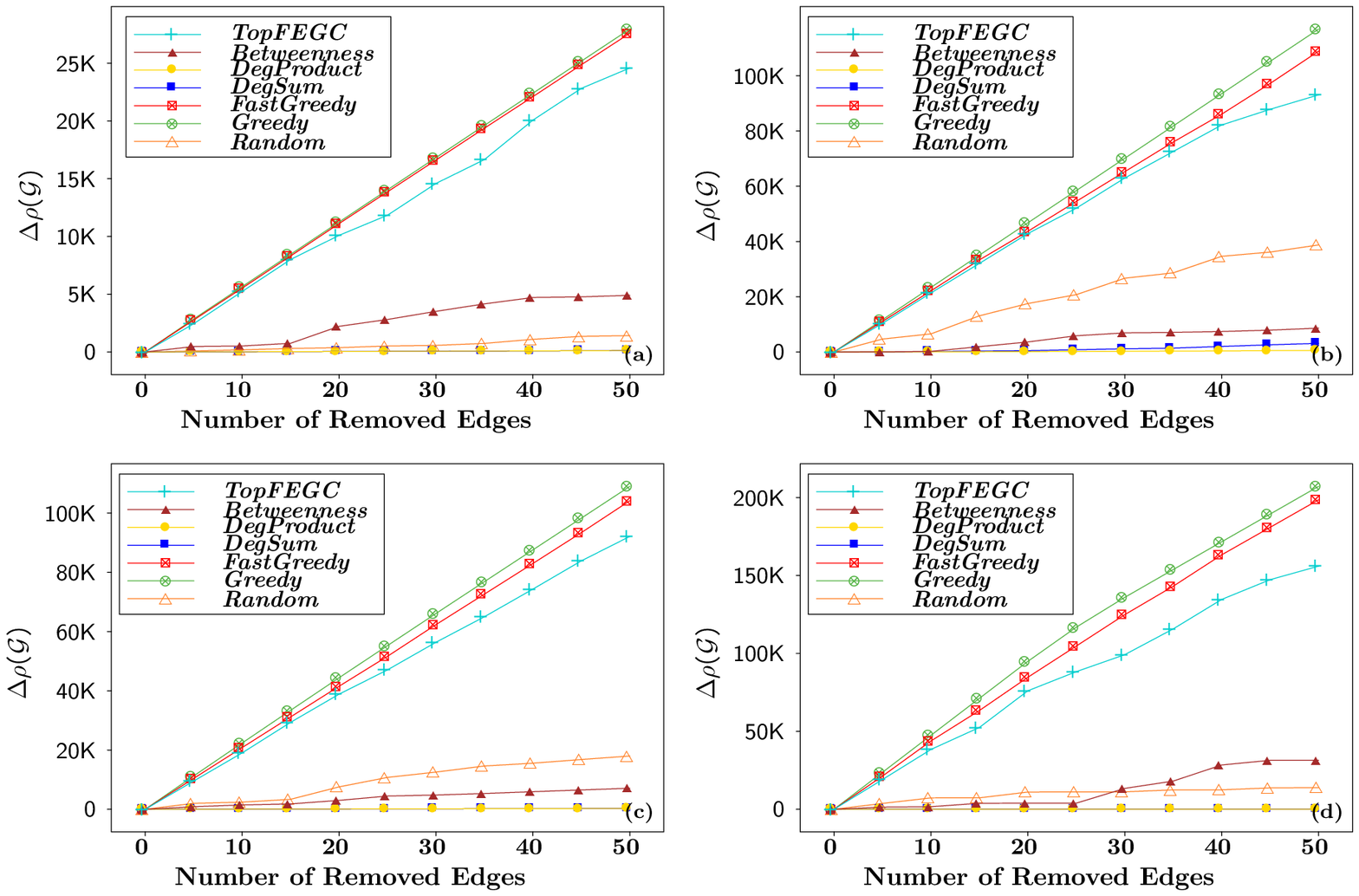}
	\caption{Increased forest index for seven methods of edge attack on datasets: (a) EmailUniv, (b) GridWorm, (c) GrQc and (d) WikiElec for increasing value of $k$.}\label{ComOpt2}
\end{figure}

\begin{figure}
	\centering
	\includegraphics[width=\linewidth]{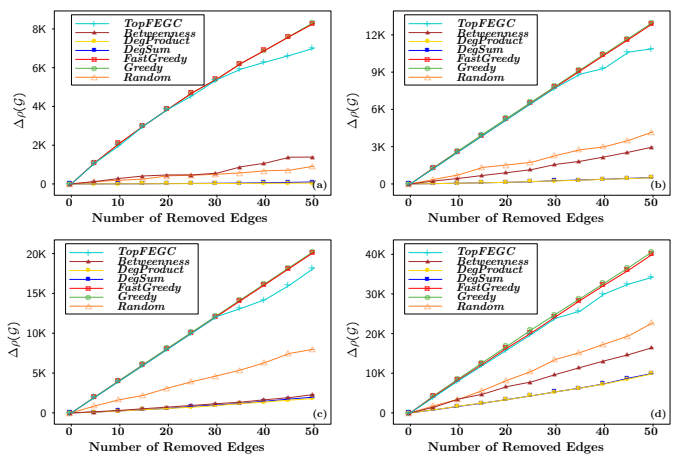}
	\caption{Increased forest index for seven methods of edge attack on datasets: (a) Crime-moreno, (b) Bio-diseasome, (c) Bio-celegans and (d) Inf-euroroad for increasing value of $k$.}\label{ComOpt3}
\end{figure}

In order to further show the effectiveness of our algorithms,  we proceed to  compare the results of our methods against the baseline schemes on eight relatively larger real-world networks: EmailUniv with $1133$ nodes and $5451$ edges, GridWorm with $3507$ nodes and $6531$ edges, GrQc with $4158$ nodes and $13422$ edges, WikiElec with $7118$ nodes and $100753$ edges, Crime-moreno with $829$ nodes and $1476$ edges, Bio-diseasome with $516$ nodes and $1188$ edges, Bio-celegans with $453$ nodes and $2025$ edges, and Inf-euroroad with $1417$ nodes and $1174$ edges.  These networks are too large to get the optimal solutions by brute-force search. For each network,  the performance of different methods for increasing value of $k$ are  displayed in  Figure~\ref{ComOpt2} and Figure~\ref{ComOpt3}. The results for these networks are impressive; we can see that, firstly, $\SpGreedy$ achieves the best performance as expected, and the proposed $\FastGreedy$  is very close to the $\SpGreedy$ method. Secondly, $\SpGreedy$ and $\FastGreedy$ consistently outperform all five
baselines and competitors, especially after the deletion of a few edges, as our two algorithms obtain much larger gain with fewer deletion, i.e., the forest index increases faster requiring less edge modifications. In addition, though TopFEGC performs close to FastGreedy for small values of $k$, the difference between \emph{TopFEGC} and \emph{FastGreedy} grows larger when $k$ increases. Finally, \emph{Betweenness}, \emph{DegProduct} and \emph{DegSum} grow flat after a certain number of edge deletion.

\subsection{Efficiency Comparison}
Although both our greedy algorithms $\SpGreedy$ and $\FastGreedy$ achieve remarkable performance in effectiveness,  we will show that $\FastGreedy$ runs much faster than $\SpGreedy$.  To this end,  we compare the efficiency of the two greedy algorithms on  a lager set of real-world networks.  For each network,  we select $k = 50$  edges  to maximize the forest index according to Algorithms~\ref{alg:SpG} and~\ref{alg:Appro}.  In Table~\ref{table:eff} we provide the results of running time and the increase of the forest index $\Delta\rho(G)$ returned by our two greedy algorithms. We observe that for moderate $\eps$, $\FastGreedy$ is significantly faster than $\SpGreedy$, with improvement becomes more significant when the graphs grow in size.  It is worth noting that $\SpGreedy$ is not applicable to the last ten networks marked with "$*$" due to the limitations of time and memory. In comparison with $\SpGreedy$, $\FastGreedy$ approximately computes $\Delta\rho(G)$ within several hours. Therefore, our algorithm $\FastGreedy$ achieves remarkable improvement in efficiency and is scalable to large networks with more than $10^6$ nodes.

We proceed to show both algorithms almost yield the same value on $\Delta\rho(G)$. We define $\beta$ and $\tilde{\beta}$ as the increase of the forest index after deleting edges selected, respectively, by $\SpGreedy$ and $\FastGreedy$. We use $\theta=|\tilde{\beta}-\beta|/\beta$ to denote relative error between $\tilde{\beta}$ and $\beta$. The results of relative
errors for different real networks and various parameter $\eps$ are
presented in Table~\ref{table:eff}. We observe that, for $\eps= 0.1, 0.2$ and $0.3$, these relative errors are negligible for all tested networks, with the largest value equal to  $5.61\%$. Thus, the results turns by $\FastGreedy$ are very close to those associated with $\SpGreedy$, implying that $\FastGreedy$ is both effective and efficient.

\section{Conclusion}

In this paper, we studied the measure and an optimization problem for network robustness in graphs, especially disconnected graphs. We first proposed to apply the forest index as a robustness measure for a connected or disconnected graph, which overcomes the weakness of existing measures. Based on the forest index, we then formulated a combinatorial optimization problem of attacking $k$ edges to maximally reduce the network robustness. We showed that the objective function of the problem is not submodular, although it is monotone. To solve the problem, we defined a novel centrality $C(S)$ of an edge set $S$, based on which we developed two approximation algorithms by deleting  $k$ edges in a greedy fashion.   The first algorithm has a proved bound of approximation factor and cubic-time complexity, while the second algorithm is nearly linear. Finally, we performed extensive experiments in various real-world networks, which show that our algorithms outperform several baseline strategies for deleting edges.

We note that our method for determining the theoretical bounds of the approximation ratio only applies to the case that the objective function is monotonic.  However, in some application scenarios, the objective function is not monotonic. As a topic for future research, it is interesting to extend or modify our method to a non-monotonic objective function. In addition, our obtained bounds for the approximation ratio depend on the eigenvalues of  graph Laplacian. Another future work is to establish a constant approximation ratio, which is independent of network parameters. Finally, it should be mentioned that although various edge-centrality based robustness measures have been compared in terms of their discriminating power~\cite{BaXuZh22}, there is no standard answer to the question---which robustness metric is prior to others, because it depends on specific applications. However, since the forest index, as a network robustness measure, can overcome those drawbacks of existing measures, we believe that it is a promising metric quantifying the robustness of a graph, especially disconnected graph.

\bibliographystyle{IEEEtran}
\bibliography{tifs}

\begin{thebibliography}{10}
\providecommand{\url}[1]{#1}
\csname url@samestyle\endcsname
\providecommand{\newblock}{\relax}
\providecommand{\bibinfo}[2]{#2}
\providecommand{\BIBentrySTDinterwordspacing}{\spaceskip=0pt\relax}
\providecommand{\BIBentryALTinterwordstretchfactor}{4}
\providecommand{\BIBentryALTinterwordspacing}{\spaceskip=\fontdimen2\font plus
\BIBentryALTinterwordstretchfactor\fontdimen3\font minus
  \fontdimen4\font\relax}
\providecommand{\BIBforeignlanguage}[2]{{%
\expandafter\ifx\csname l@#1\endcsname\relax
\typeout{** WARNING: IEEEtran.bst: No hyphenation pattern has been}%
\typeout{** loaded for the language `#1'. Using the pattern for}%
\typeout{** the default language instead.}%
\else
\language=\csname l@#1\endcsname
\fi
#2}}
\providecommand{\BIBdecl}{\relax}
\BIBdecl

\bibitem{WaFeKoMa19}
X.~Wang, L.~Feng, R.~E. Kooij, and J.~L. Marzo, ``Inconsistencies among
  spectral robustness metrics,'' in \emph{Qual. Rel. Secur. Robustness
  Heterogeneous Syst.}\hskip 1em plus 0.5em minus 0.4em\relax Springer
  International Publishing, 2019, pp. 119--136.

\bibitem{HiBaSa09}
P.~Hines, K.~Balasubramaniam, and E.~C. Sanchez, ``Cascading failures in power
  grids,'' \emph{IEEE Potentials}, vol.~28, no.~5, pp. 24--30, 2009.

\bibitem{Tr21}
T.~Reshmi, ``Information security breaches due to ransomware attacks - a
  systematic literature review,'' \emph{Int. J. Inf. Manage.}, vol.~1, no.~2,
  p. 100013, 2021.

\bibitem{ElKo13}
W.~Ellens and R.~E. Kooij, ``Graph measures and network robustness,''
  \emph{arXiv preprint arXiv:1311.5064}, 2013.

\bibitem{OeFa21}
M.~Oehlers and B.~Fabian, ``Graph metrics for network robustness—a survey,''
  \emph{Mathematics}, vol.~9, no.~8, p. 895, 2021.

\bibitem{FiClFeViDo04}
F.~Radicchi, C.~Castellano, F.~Cecconi, V.~Loreto, and D.~Parisi, ``Defining
  and identifying communities in networks,'' \emph{Proc. Natl. Acad. Sci.
  U.S.A.}, vol. 101, no.~9, pp. 2658--2663, feb 2004.

\bibitem{Ne05}
M.~E.~J. Newman, ``A measure of betweenness centrality based on random walks,''
  \emph{Social Netw.}, vol.~27, no.~1, pp. 39--54, 2005.

\bibitem{BrFl05}
U.~Brandes and D.~Fleischer, ``Centrality measures based on current flow,'' in
  \emph{Proc. 22nd Annu. Symp. Theor. Aspects Comput. Sci.}, 2005, pp.
  533--544.

\bibitem{KoIsBa15}
I.~Kov{\'a}cs and A.~Barab{\'a}si, ``Network science: Destruction perfected,''
  \emph{Nature}, vol. 524, no. 7563, pp. 38--39, 2015.

\bibitem{AzGaNa17}
H.~Aziz, S.~Gaspers, and K.~Najeebullah, ``Weakening covert networks by
  minimizing inverse geodesic length,'' in \emph{Proc. 26th Int. Joint Conf.
  Artif. Intell.}\hskip 1em plus 0.5em minus 0.4em\relax AAAI Press, 2017, p.
  779–785.

\bibitem{ZhAdSaVuPr16}
Y.~Zhang, A.~Adiga, S.~Saha, A.~Vullikanti, and B.~A. Prakash, ``Near-optimal
  algorithms for controlling propagation at group scale on networks,''
  \emph{IEEE Trans. Knowl. Data En.}, vol.~28, no.~12, p. 3339–3352, 2016.

\bibitem{GaMaTo05}
A.~Ganesh, L.~Massouli{\'e}, and D.~Towsley, ``The effect of network topology
  on the spread of epidemics,'' in \emph{Proc. 24th Annu. Joint Conf. IEEE
  Comput. Commun. Soc.}, vol.~2.\hskip 1em plus 0.5em minus 0.4em\relax IEEE,
  2005, pp. 1455--1466.

\bibitem{PaCaVaVe15}
R.~Pastor-Satorras, C.~Castellano, P.~Van~Mieghem, and A.~Vespignani,
  ``Epidemic processes in complex networks,'' \emph{Rev. Mod. Phys.}, vol.~87,
  no.~3, p. 925, 2015.

\bibitem{HaZh17}
Y.~Hayel and Q.~Zhu, ``Epidemic protection over heterogeneous networks using
  evolutionary poisson games,'' \emph{IEEE Trans. Inf. Forensics Security},
  vol.~12, no.~8, pp. 1786--1800, 2017.

\bibitem{KeWh91}
J.~Kephart and S.~White, ``Directed-graph epidemiological models of computer
  viruses,'' in \emph{Proc. IEEE Comput. Soc. Symp. Res. Secur. Privacy}, 1991,
  pp. 343--343.

\bibitem{KeWh93}
J.~O. Kephart and S.~R. White, ``Measuring and modeling computer virus
  prevalence,'' in \emph{Proc. IEEE Comput. Soc. Symp. Res. Secur.
  Privacy}.\hskip 1em plus 0.5em minus 0.4em\relax IEEE, 1993, pp. 2--15.

\bibitem{JaCr12}
J.~T. Jackson and S.~Creese, ``Virus propagation in heterogeneous bluetooth
  networks with human behaviors,'' \emph{IEEE Trans. Dependable and Secure
  Comput.}, vol.~9, no.~6, pp. 930--943, 2012.

\bibitem{WaNiZhLiNi16}
X.~Wang, W.~Ni, K.~Zheng, R.~P. Liu, and X.~Niu, ``Virus propagation modeling
  and convergence analysis in large-scale networks,'' \emph{IEEE Trans. Inf.
  Forensics Security}, vol.~11, no.~10, pp. 2241--2254, 2016.

\bibitem{YiShPa22}
Y.~Yi, L.~Shan, P.~E. Par{\'e}, and K.~H. Johansson, ``Edge deletion algorithms
  for minimizing spread in sir epidemic models,'' \emph{SIAM J. Control
  Optim.}, vol.~60, no.~2, pp. S246--S273, 2022.

\bibitem{DaHa07}
E.~M. Daly and M.~Haahr, ``Social network analysis for routing in disconnected
  delay-tolerant manets,'' in \emph{Proc. 8th ACM Int. Symp. Ad hoc Netw.
  Comput.}\hskip 1em plus 0.5em minus 0.4em\relax ACM, 2007, pp. 32--40.

\bibitem{HuLiWu18}
L.~Huang, L.~Liao, and C.~H. Wu, ``Completing sparse and disconnected
  protein-protein network by deep learning,'' \emph{BMC Bioinf.}, vol.~19,
  no.~1, p. 103, 2018.

\bibitem{ChSh97}
P.~Y. Chebotarev and E.~V. Shamis, ``The matrix-forest theorem and measuring
  relations in small social groups,'' \emph{Automat. Remote Control}, vol.~58,
  no.~9, pp. 1505--1514, 1997.

\bibitem{ChSh98}
------, ``On proximity measures for graph vertices,'' \emph{Automat. Remote
  Control}, vol.~59, no.~10, pp. 1443--1459, 1998.

\bibitem{Me98}
R.~Merris, ``{Doubly stochastic graph matrices, II},'' \emph{Linear Multilinear
  Algebr.}, vol.~45, no. 2-3, pp. 275--285, 1998.

\bibitem{Ch08}
P.~Chebotarev, ``Spanning forests and the golden ratio,'' \emph{Discrete Appl.
  Math.}, vol. 156, no.~5, pp. 813--821, 2008.

\bibitem{ToPrTsElFaCh10}
H.~Tong, B.~A. Prakash, C.~Tsourakakis, T.~Eliassi-Rad, C.~Faloutsos, and D.~H.
  Chau, ``On the vulnerability of large graphs,'' in \emph{Proc. IEEE Int.
  Conf. Data Mining}, 2010, pp. 1091--1096.

\bibitem{FrYaKu21}
S.~Freitas, D.~Yang, S.~Kumar, H.~Tong, and D.~H. Chau, ``Graph vulnerability
  and robustness: {A} survey,'' \emph{IEEE Trans. Knowl. Data Eng.}, vol.~35,
  no.~6, pp. 5915--59\,341, 2023.

\bibitem{NgEf06}
A.~K. Ng and J.~Efstathiou, ``Structural robustness of complex networks,''
  \emph{Phys. Rev.}, vol.~3, pp. 175--188, 2006.

\bibitem{Li78}
L.~C. Freeman, ``Centrality in social networks conceptual clarification,''
  \emph{Social Netw.}, vol.~1, no.~3, pp. 215--239, 1978.

\bibitem{WeGuViMe04}
G.~Weichenberg, V.~W. Chan, and M.~M{\'e}dard, ``High-reliability topological
  architectures for networks under stress,'' \emph{IEEE J. Sel. Area. Comm.},
  vol.~22, no.~9, pp. 1830--1845, 2004.

\bibitem{BaHo09}
J.~S. Baras and P.~Hovareshti, ``Efficient and robust communication topologies
  for distributed decision making in networked systems,'' in \emph{Proc. 48h
  IEEE Conf. Decis. Control}.\hskip 1em plus 0.5em minus 0.4em\relax IEEE,
  2009, pp. 3751--3756.

\bibitem{GhBoSa08}
A.~Ghosh, S.~Boyd, and A.~Saberi, ``Minimizing effective resistance of a
  graph,'' \emph{SIAM Rev.}, vol.~50, no.~1, pp. 37--66, 2008.

\bibitem{LiZh18}
H.~Li and Z.~Zhang, ``Kirchhoff index as a measure of edge centrality in
  weighted networks: {N}early linear time algorithms,'' in \emph{Proc. 29th
  Annu. ACM-SIAM Symp. Discrete Algorithms}.\hskip 1em plus 0.5em minus
  0.4em\relax SIAM, 2018, pp. 2377--2396.

\bibitem{YaMoQiWe19}
C.~Yang, J.~Mao, X.~Qian, and P.~Wei, ``Designing robust air transportation
  networks via minimizing total effective resistance,'' \emph{IEEE Trans.
  Intell. Transp. Syst.}, vol.~20, no.~6, pp. 2353--2366, 2019.

\bibitem{LiChZh17}
Y.~Lin, W.~Chen, and Z.~Zhang, ``Assessing percolation threshold based on
  high-order non-backtracking matrices,'' in \emph{Proc. 26th Int. Conf. World
  Wide Web}, 2017, pp. 223--232.

\bibitem{LiZh19}
Y.~Lin and Z.~Zhang, ``Non-backtracking centrality based random walk on
  networks,'' \emph{Comput. J.}, vol.~62, no.~1, pp. 63--80, 2019.

\bibitem{ZhZhCh21}
Z.~Zhang, Z.~Zhang, and G.~Chen, ``Minimizing spectral radius of
  non-backtracking matrix by edge removal,'' in \emph{Proc. 30th ACM Conf. Inf.
  Knowl. Manage.}\hskip 1em plus 0.5em minus 0.4em\relax ACM, 2021, p.
  2657–2667.

\bibitem{DeMeRoSaVa18}
P.~De~Meo, F.~Messina, D.~Rosaci, G.~M.~L. Sarné, and A.~V. Vasilakos,
  ``Estimating graph robustness through the randic index,'' \emph{IEEE Trans.
  Cybern.}, vol.~48, no.~11, pp. 3232--3242, 2018.

\bibitem{BoXi19}
B.~Ning and X.~Peng, ``The randic index and signless laplacian spectral radius
  of graphs,'' \emph{Discrete Math.}, vol. 342, no.~3, pp. 643--653, 2019.

\bibitem{KiKiLaPh16}
R.~K. Kincaid, S.~J. Kunkler, M.~D. Lamar, and D.~J. Phillips, ``Algorithms and
  complexity results for finding graphs with extremal randic index,''
  \emph{Networks}, vol.~67, no.~4, pp. 338--347, 2016.

\bibitem{XuShZhKaZh20}
W.~Xu, Y.~Sheng, Z.~Zhang, H.~Kan, and Z.~Zhang, ``Power-law graphs have
  minimal scaling of {K}emeny constant for random walks,'' in \emph{Proc. 29th
  Int. Conf. World Wide Web}, 2020, pp. 46--56.

\bibitem{ZhXuZh20}
Z.~Zhang, W.~Xu, and Z.~Zhang, ``Nearly linear time algorithm for mean hitting
  times of random walks on a graph,'' in \emph{Proc. 13th Int. Conf. Web Search
  Data Mining}, 2020, pp. 726--734.

\bibitem{LiWaBuCaSh21}
H.-J. Li, L.~Wang, Z.~Bu, J.~Cao, and Y.~Shi, ``Measuring the network
  vulnerability based on {M}arkov criticality,'' \emph{ACM Trans. Knowl.
  Discovery Data}, vol.~16, no.~2, pp. 1--24, 2021.

\bibitem{TsSuLiHsMy94}
F.-S.~P. Tsen, T.-Y. Sung, M.-Y. Lin, L.-H. Hsu, and W.~Myrvold, ``Finding the
  most vital edges with respect to the number of spanning trees,'' \emph{IEEE
  Trans. Rel.}, vol.~43, no.~4, pp. 600--603, 1994.

\bibitem{Ra98}
V.~B. Rao, ``Most-vital edge of a graph with respect to spanning trees,''
  \emph{IEEE Trans. Rel.}, vol.~47, no.~1, pp. 6--7, 1998.

\bibitem{VaStKuLiVaLiWa11}
P.~Van~Mieghem, D.~Stevanovi{\'c}, F.~Kuipers, C.~Li, R.~Van De~Bovenkamp,
  D.~Liu, and H.~Wang, ``Decreasing the spectral radius of a graph by link
  removals,'' \emph{Phys. Rev. E}, vol.~84, no.~1, p. 016101, 2011.

\bibitem{ToPrElFaFa12}
H.~Tong, B.~A. Prakash, T.~Eliassi-Rad, M.~Faloutsos, and C.~Faloutsos,
  ``Gelling, and melting, large graphs by edge manipulation,'' in \emph{Proc.
  21st Int. Conf. Inf. Knowl. Man.}\hskip 1em plus 0.5em minus 0.4em\relax ACM,
  2012, pp. 245--254.

\bibitem{GaNa19}
S.~Gaspers and K.~Najeebullah, ``Optimal surveillance of covert networks by
  minimizing inverse geodesic length,'' in \emph{Proc. 33rd Int. Joint Conf.
  Artif. Intell.}, vol.~33, no.~01, 2019, pp. 533--540.

\bibitem{ScBoVa87}
A.~A. Schoone, H.~L. Bodlaender, and J.~Van~Leeuwen, ``Diameter increase caused
  by edge deletion,'' \emph{J. Graph Theory}, vol.~11, no.~3, pp. 409--427,
  1987.

\bibitem{LiShHu00}
W.~Liang, X.~Shen, and Q.~Hu, ``Finding the most vital edge for graph
  minimization problems on meshes and hypercubes,'' \emph{Int. J. Paral.
  Distrib. Syst. Netw.}, vol.~3, no.~4, pp. 197--205, 2000.

\bibitem{MeMaSiSi20}
S.~Medya, T.~Ma, A.~Silva, and A.~Singh, ``A game theoretic approach for k-core
  minimization,'' in \emph{Proc. 19th Int. Conf. Auton. Agents MultiAgent
  Syst.}, 2020.

\bibitem{DiXuThPaZn12}
T.~N. Dinh, Y.~Xuan, M.~T. Thai, P.~M. Pardalos, and T.~Znati, ``{On new
  approaches of assessing network vulnerability: Hardness and approximation},''
  \emph{IEEE/ACM Trans. Netw.}, vol.~20, no.~2, pp. 609--619, 2012.

\bibitem{ShNgXuTh13}
Y.~Shen, N.~P. Nguyen, Y.~Xuan, and M.~T. Thai, ``On the discovery of critical
  links and nodes for assessing network vulnerability,'' \emph{IEEE/ACM Trans.
  Netw.}, vol.~21, no.~3, pp. 963--973, 2013.

\bibitem{VerPrPa19}
A.~Veremyev, O.~A. Prokopyev, and E.~L. Pasiliao, ``Finding critical links for
  closeness centrality,'' \emph{INFORMS J. Comput.}, vol.~31, no.~2, pp.
  367--389, 2019.

\bibitem{Me19}
S.~Medya, \emph{Scalable Algorithms for Network Design}.\hskip 1em plus 0.5em
  minus 0.4em\relax University of California, Santa Barbara, 2019.

\bibitem{ChPeYiTo}
C.~Chen, R.~Peng, L.~Ying, and H.~Tong, ``Network connectivity optimization:
  Fundamental limits and effective algorithms,'' in \emph{Proc. 24th ACM SIGKDD
  Int. Conf. Knowl. Discovery Data Mining}.\hskip 1em plus 0.5em minus
  0.4em\relax New York, NY, USA: Association for Computing Machinery, 2018, p.
  1167–1176.

\bibitem{Ch18}
C.~Chen, ``Connectivity in complex networks: Measures, inference and
  optimization,'' Arizona State University, Tech. Rep., 2019.

\bibitem{GhBo06}
A.~Ghosh and S.~Boyd, ``Growing well-connected graphs,'' in \emph{Proc. 45th
  IEEE Conf. Decis. Control}, 2006, pp. 6605--6611.

\bibitem{LiPaYiZh20}
H.~Li, S.~Patterson, Y.~Yi, and Z.~Zhang, ``Maximizing the number of spanning
  trees in a connected graph,'' \emph{IEEE Trans. Inf. Theory}, vol.~66, no.~2,
  pp. 1248--1260, 2020.

\bibitem{SoArAmPrAn18}
S.~Medya, A.~Silva, A.~Singh, P.~Basu, and A.~Swami, ``Group centrality
  maximization via network design,'' \emph{Proc. SIAM Int. Conf. Data Mining},
  pp. 126--134, 2018.

\bibitem{ChToPrElFaFa16}
C.~Chen, H.~Tong, B.~A. Prakash, T.~Eliassi-Rad, M.~Faloutsos, and
  C.~Faloutsos, ``Eigen-optimization on large graphs by edge manipulation,''
  \emph{ACM Trans. Knowl. Discov. Data}, vol.~10, no.~4, 2016.

\bibitem{ChXuLeDuTaChPr17}
L.~Chen, X.~Xu, S.~Lee, S.~Duan, A.~G. Tarditi, S.~Chinthavali, and B.~A.
  Prakash, ``Hotspots: Failure cascades on heterogeneous critical
  infrastructure networks,'' in \emph{Proc. 26th ACM Conf. Inf. Knowl.
  Manage.}\hskip 1em plus 0.5em minus 0.4em\relax New York, NY, USA:
  Association for Computing Machinery, 2017, p. 1599–1607.

\bibitem{ChPeYiTo21}
C.~Chen, R.~Peng, L.~Ying, and H.~Tong, ``Fast connectivity minimization on
  large-scale networks,'' \emph{ACM Trans. Knowl. Discov. Data}, vol.~15,
  no.~3, 2021.

\bibitem{ChHeBlTo17}
C.~Chen, J.~He, N.~Bliss, and H.~Tong, ``Towards optimal connectivity on
  multi-layered networks,'' \emph{IEEE Trans. Knowl. Data En.}, vol.~29,
  no.~10, pp. 2332--2346, 2017.

\bibitem{SpSr11}
D.~A. Spielman and N.~Srivastava, ``Graph sparsification by effective
  resistances,'' \emph{{SIAM} J. Comput.}, vol.~40, no.~6, pp. 1913--1926,
  2011.

\bibitem{LiSc18}
H.~Li and A.~Schild, ``Spectral subspace sparsification,'' in \emph{Proc. 59nd
  IEEE Symp. Found. Comput. Sci.}\hskip 1em plus 0.5em minus 0.4em\relax IEEE,
  2018, pp. 385--396.

\bibitem{NeWoFi78}
G.~L. Nemhauser, L.~A. Wolsey, and M.~L. Fisher, ``An analysis of
  approximations for maximizing submodular set functions - {I},'' \emph{Math.
  Program.}, vol.~14, no.~1, pp. 265--294, 1978.

\bibitem{BiBuKrTs17}
A.~A. Bian, J.~M. Buhmann, A.~Krause, and S.~Tschiatschek, ``Guarantees for
  greedy maximization of non-submodular functions with applications,'' in
  \emph{Proc. 34th Int. Conf. Machine Learn.}\hskip 1em plus 0.5em minus
  0.4em\relax JMLR. org, 2017, pp. 498--507.

\bibitem{Me97}
R.~Merris, ``Doubly stochastic graph matrices,'' \emph{Univ. Beograd. Publ.
  Elektrotehn. Fak. Ser. Mat.}, vol.~1, no.~8, pp. 64--71, 1997.

\bibitem{Me73}
C.~D. Meyer, Jr, ``Generalized inversion of modified matrices,'' \emph{SIAM J.
  Appl. Math.}, vol.~24, no.~3, pp. 315--323, 1973.

\bibitem{JoLi84}
W.~B. Johnson and J.~Lindenstrauss, ``{Extensions of Lipschitz mappings into a
  Hilbert space},'' \emph{Contemporary Math.}, vol.~26, pp. 189--206, 1984.

\bibitem{Ac01}
D.~Achlioptas, ``Database-friendly random projections,'' in \emph{Proc. 20th
  ACM SIGMOD-SIGACT-SIGART Symp. Princ. Syst.}\hskip 1em plus 0.5em minus
  0.4em\relax ACM, 2001, pp. 274--281.

\bibitem{SpTe14}
D.~Spielman and S.~Teng, ``Nearly linear time algorithms for preconditioning
  and solving symmetric, diagonally dominant linear systems,'' \emph{{SIAM} J.
  Matrix Anal. Appl.}, vol.~35, no.~3, pp. 835--885, 2014.

\bibitem{Ku13}
J.~Kunegis, ``Konect: The koblenz network collection,'' in \emph{Proc. 22th
  Int. Conf. World Wide Web}.\hskip 1em plus 0.5em minus 0.4em\relax New York,
  USA: ACM, 2013, pp. 1343--1350.

\bibitem{LeSo16}
J.~Leskovec and R.~Sosi{\v{c}}, ``{SNAP: A general-purpose network analysis and
  graph-mining library},'' \emph{ACM Trans. Intel. Syst. Tech.}, vol.~8, no.~1,
  p.~1, 2016.

\bibitem{KySa16}
R.~Kyng and S.~Sachdeva, ``{Approximate Gaussian elimination for
  Laplacians-fast, sparse, and simple},'' in \emph{Proc. 57th IEEE Symp. Found.
  Comput. Sci.}\hskip 1em plus 0.5em minus 0.4em\relax IEEE, 2016, pp.
  573--582.

\bibitem{Br01}
U.~Brandes, ``A faster algorithm for betweenness centrality,'' \emph{J. Math.
  Sociol.}, vol.~25, no.~2, pp. 163--177, 2001.

\bibitem{BaXuZh22}
Q.~Bao, W.~Xu, and Z.~Zhang, ``Benchmark for discriminating power of edge
  centrality metrics,'' \emph{Comput. J.}, vol.~65, no.~12, pp. 3141--3155,
  2022.

\end{thebibliography}

\begin{IEEEbiography}{Liwang Zhu}
		received the B.Eng. degree in computer science and technology, Nanjing University of Science and Technology, Nanjing, China, in 2018. He is currently pursuing the Ph.D. degree in the School of Computer Science, Fudan University, Shanghai, China. His research interests include social networks, graph data mining, and network science.
\end{IEEEbiography}
	
\begin{IEEEbiography}{Qi Bao}
		received the B.Sc. degree in School of Computer Science, Fudan University, Shanghai, China, in 2017, where he is currently pursuing the Ph.D. degree.
		His research interests include graph algorithms, social networks, and network science.
\end{IEEEbiography}

\begin{IEEEbiography}{Zhongzhi Zhang}
	(M'19)	 received the B.Sc. degree in applied mathematics from Anhui University, Hefei, China, in 1997 and the Ph.D. degree in management science and engineering from Dalian University of Technology, Dalian, China, in 2006. \\
	From 2006 to 2008, he was a Post-Doctoral Research Fellow with Fudan University, Shanghai, China, where he is currently a Full Professor with the School of Computer Science. He has published over 160 papers in international journals or conferences. 
 He was selected as one of the most cited Chinese researchers
	(Elsevier) in 2019,  2020, 2021, and 2022. His current research interests include network science, graph data mining, social network analysis, computational social science, spectral graph theory, and random walks. \\
	Dr. Zhang was a recipient of the Excellent Doctoral Dissertation Award of Liaoning Province, China, in 2007, the Excellent Post-Doctor Award of Fudan University in 2008, the Shanghai Natural Science Award (third class) in 2013, the CCF Natural Science Award (second class) in 2022, and the Wilkes Award for the best paper published in The Computer Journal in 2019. He is a member of the IEEE.
\end{IEEEbiography}


\end{document}